\documentclass{article}[11pt]

\usepackage[pagebackref,colorlinks=true,linkcolor=blue,urlcolor=blue,citecolor=blue,pdfstartview=FitH,breaklinks]{hyperref}
\usepackage{amsmath}
\usepackage{amsfonts}
\usepackage{amsthm} 
\usepackage{thmtools} 
\usepackage{enumitem} 
\usepackage{multirow}
\usepackage[margin=1in]{geometry}
\usepackage{mathtools}
\usepackage[capitalise,nameinlink]{cleveref}
\usepackage{url}
\usepackage{complexity}
\usepackage{tcolorbox}
\usepackage{xspace}

\usepackage{pgf,tikz,pgfplots}
\pgfplotsset{compat=1.14}
\usepackage{mathrsfs}
\usetikzlibrary{arrows}

\declaretheorem[name=Theorem,numberwithin=section]{theorem}
\declaretheorem[name=Corollary,sibling=theorem]{corollary}
\declaretheorem[name=Claim,sibling=theorem]{claim}
\declaretheorem[name=Fact,sibling=theorem]{fact}
\declaretheorem[name=Lemma,sibling=theorem]{lemma}
\declaretheorem[name=Definition, sibling=theorem]{definition}
\declaretheorem[name=Remark, sibling=theorem]{remark}

\theoremstyle{definition}
\declaretheorem[name=Algorithm, sibling=theorem]{algorithm}
\declaretheorem[name=Algorithm, numbered=no]{algorithmNoNum}

\newcommand{\makeparenthesis}[3]{%
	\expandafter\renewcommand\csname #1\endcsname[1]{\mathchoice
		{{\left#2##1\right#3}}%
		{{\left#2##1\right#3}}%
		{{\left#2##1\right#3}}%
		{{\left#2##1\right#3}}%
	}%
}
\newcommand{\pars}[1]{} 
\makeparenthesis{pars}{\lparen}{\rparen}
\newcommand{\brackets}[1]{} 
\makeparenthesis{brackets}{\lbrack}{\rbrack}
\newcommand{\braces}[1]{} 
\makeparenthesis{braces}{\lbrace}{\rbrace}
\newcommand{\abs}[1]{} 
\makeparenthesis{abs}{|}{|}
\newcommand{\ceil}[1]{} 
\makeparenthesis{ceil}{\lceil}{\rceil}
\newcommand{\floor}[1]{} 
\makeparenthesis{floor}{\lfloor}{\rfloor}
\newcommand{\iprod}[1]{} 
\makeparenthesis{iprod}{\langle}{\rangle}

\newcommand{\N}{\mathbb{N}}
\newcommand{\Z}{\mathbb{Z}}
\newcommand{\bool}{\braces{0,1}}
\newcommand{\Exp}{\mathbb{E}}

\newfunc{\Maj}{Maj}
\newcommand{\restrict}[1]{{|_{#1}}}
\newcommand{\F}{\mathbb{F}}
\newcommand{\calL}{\mathcal{L}}
\newcommand{\agla}{\Delta}

\newcommand{\row}{{\mathrm{row}}}
\newcommand{\column}{{\mathrm{col}}}
\newcommand{\shared}{{\mathrm{shared}}}
\newcommand{\sharedrow}{{\mathrm{shared.row}}}
\newcommand{\sharedcolumn}{{\mathrm{shared.col}}}
\newcommand{\rop}{\mathrm{ROP}}
\newcommand{\RNL}{{\mathrm{RNL}}}

\renewcommand{\outer}{\mathrm{out}} 
\newcommand{\inner}{\mathrm{in}}
\newcommand{\comp}{\mathrm{comp}}

\newcommand{\compr}[1]{\hat{#1}}

\DeclareMathOperator{\numones}{\#_1}

\newcommand{\Enc}{\mathrm{Enc}}

\newcommand{\Dec}{\mathrm{Dec}}
\newcommand{\Test}{\mathrm{Test}}

\newlang{\threeSAT}{3SAT}
\newlang{\threeLIN}{3LIN}

\DeclareMathOperator{\val}{val}

\newcommand{\etal}{\textit{et al}.\xspace}
\newcommand{\defemph}[1]{{\sf #1}}

\newcommand{\getsr}{\mathbin{\stackrel{\mbox{\tiny R}}{\gets}}}

\newcommand{\pintegers}{\N} 

\newcommand{\tdelim}{\tabularnewline \hline}

\newcommand{\bset}{S}
\newcommand{\shift}{\mathsf{shift}}

\usepackage{etoolbox}
\makeatletter
\patchcmd{\@bibitem}{\ignorespaces}{\label{bib-#1}\ignorespaces}{}{}
\makeatother
\newcommand{\myqcite}[2]{[\ref{bib-#2}, #1]}
\newcommand{\mytqcite}[2]{\texorpdfstring{\cite[#1]{#2}}{\myqcite{#1}{#2}}}

\usepackage[colorinlistoftodos]{todonotes}

\usepackage{pgf,tikz,pgfplots}
\usepackage{mathrsfs}
\usetikzlibrary{arrows}
\usetikzlibrary{patterns}
\usetikzlibrary{decorations.pathmorphing}

\renewcommand{\epsilon}{\varepsilon}
\newcommand{\maxcut}{\ComplexityFont{MAXCUT}}

\newcommand{\calC}{{\mathcal{C}}}

\title{Rigid Matrices From Rectangular PCPs\\ {\Large {\em or: Hard Claims Have Complex Proofs }}\thanks{A conference version of this paper appeared in {\em Proc. 61st FOCS}, 2020~\cite{BhangaleHPT2020}.}}

\author{
    Amey Bhangale \thanks{University of California, Riverside, CA, USA. \protect\url{ameyb@ucr.edu}.}
    \thanks{Work done while the first, second and third author were participating in the \emph{Proofs, Consensus, and Decentralizing Society} program at the Simons Institute for the Theory of Computing, Berkeley, CA, USA.}\\
	\and 
	Prahladh Harsha \thanks{Tata Institute of Fundamental Research, Mumbai, India. \protect\url{prahladh@tifr.res.in}. Research supported by the Department of Atomic Energy, Government of India, under project no. 12-R\&D-TFR-5.01-0500 and in part by the Swarnajayanti fellowship.}
	\footnotemark[2]
	\and 
	Orr Paradise \thanks{University of California, Berkeley, CA, USA. \protect\url{orrp@eecs.berkeley.edu}.}
		\footnotemark[2]
	\and 
	Avishay Tal \thanks{University of California, Berkeley, CA, USA. \protect\url{atal@berkeley.edu}.}
}

\begin{document}

\maketitle

\begin{abstract} 

We introduce a variant of PCPs, that we refer to as \emph{rectangular} PCPs, wherein proofs are thought of as square matrices, and the random coins used by the verifier can be partitioned into two disjoint sets, one determining the \emph{row} of each query and the other determining the \emph{column}.

We construct PCPs that are \emph{efficient}, \emph{short}, \emph{smooth} and (almost-)\emph{rectangular}. As a key application, we show that proofs for hard languages in $\mathsf{NTIME}(2^n)$, when viewed as matrices, are rigid infinitely often. This strengthens and simplifies a recent result of Alman and Chen [FOCS, 2019] constructing explicit rigid matrices in $\mathsf{FNP}$. Namely, we prove the following theorem:
\begin{itemize}
    \item There is a constant $\delta \in (0,1)$ such that there is an 
      $\mathsf{FNP}$-machine that, for infinitely many $N$, on input
      $1^N$ outputs $N \times N$ matrices with entries in $\mathbb{F}_2$ that are $\delta N^2$-far
      (in Hamming distance) from matrices of rank at most $2^{\log  N/\Omega(\log \log N)}$.
\end{itemize}

Our construction of rectangular PCPs starts with an analysis of how randomness yields queries in the Reed--Muller-based outer PCP of Ben-Sasson, Goldreich, Harsha, Sudan and Vadhan [SICOMP, 2006; CCC, 2005]. We then show how to preserve rectangularity under PCP composition and a smoothness-inducing transformation. This warrants refined and stronger notions of rectangularity, which we prove for the outer PCP and its transforms.

\end{abstract}

\newpage

\section{Introduction}\label{sec:intro}
\newcommand{\rg}{\Delta}

An $N \times N$ matrix with entries from a field $\F$ is said to be {\em $(\rg, \rho)$-rigid} 
if its Hamming distance from the set of $N\times N$ matrices of rank at most $\rho$ is greater than $\rg$. In other words, an $(\rg,\rho)$-rigid matrix is a matrix that cannot be expressed as a sum of two matrices, $L+S$, where $\mathrm{rank}(L)\le \rho$ and $S$ has at most $\rg$ non-zero entries.
For concreteness, this work focuses on rigidity with respect to the field $\F_2$.

Constructing rigid matrices has been a long-standing open problem in computational complexity theory since their introduction by Valiant more than four decades ago \cite{Valiant1977}. Valiant showed that for any $(N^{1+\varepsilon}, N/\log\log(N))$-rigid matrix, evaluating the corresponding linear transformation requires circuits of either super-linear size or super-logarithmic depth. Thus, an explicit construction of a such matrices gives explicit problems that cannot be solved in linear-size logarithmic-depth circuits.

Razborov~\cite{Razborov1989} (see also \cite{Wunderlich2012}) considered the other end of the spectrum of parameters, in which the distance $\rg$ is quite high but the rank $\rho$ is much smaller; namely, $\rg = \delta \cdot N^2$ for constant $\delta > 0$ and $\rho = 2^{(\log\log N)^{\omega(1)}}$. Razborov showed that strongly-explicit matrices\footnote{An $N \times N$ matrix $A$ is strongly-explicit if, given $i,j \in N$, one can compute $A_{i,j}$ in $\polylog(N)$ time.} with these rigidity parameters imply a lower bound for the communication-complexity analog of the Polynomial Hierarchy,  $\PH^{cc}$.

In other words, while Valiant's regime focuses on very high rank $\rho$, Razborov's focuses on very high distance $\rg$. Achieving lower bounds for either
$\PH^{cc}$ (via Razborov's reduction) or linear-size log-depth circuits (via Valiant's reduction) are two central long-standing open questions in complexity theory.

Despite a lot of effort, state of the art results on matrix rigidity fall short of solving both Valiant and Razborov's challenges. The current best poly-time constructions yields $(\frac{N^2}{\rho} \log(\frac{N}{\rho}), \rho)$-rigid matrices, for any parameter $\rho$ (see Friedman~\cite{Friedman1993}; Shokrollahi, Spielman, and Stemann~\cite{ShokrollahiSS1997}).
Goldreich and Tal~\cite{GoldreichT2018} gave a randomized poly-time algorithm that uses $O(N)$ random bits and produces an $N \times N$ matrix that is $(\frac{N^3}{\rho^2 \log N}, \rho)$-rigid for any parameter $\rho\ge \sqrt{N}$, with high probability.%
\footnote{Despite its ``semi-explicitness'', if this construction obtains Valiant's rigidity parameters, then lower bounds would be implied, since one can take the randomness to be part of the input, yielding a (related) explicit problem that has no linear-size logarithmic-depth circuits.}
Recent breakthrough results \cite{AlmanW2017,DvirE2019,DvirL2019} showed that several long-standing candidate construction of explicit matrices, like the Hadamard or the FFT matrices, are less rigid than previously believed, and in particular do not meet Valiant's challenge.

Most previous attempts were of combinatorial or algebraic nature (see a survey of Lokam~\cite{Lokam2009}). In contrast to these, a recent remarkable work of Alman and Chen \cite{AlmanC2019} proposes a novel approach that uses ideas from complexity theory to construct rigid matrices in $\FNP$:%
\footnote{The complexity class $\FNP$ is the function-problem extension of the decision-problem class $\NP$. Formally, a relation $R(x,y)$ is in $\FNP$ if there exists a non-deterministic polynomial-time Turing machine $M$ such that for any input $x$, $M\pars{x}$ outputs $y$ such $R(x,y)=1$ or rejects if no such $y$ exists.}

\begin{theorem}[\cite{AlmanC2019}]\label{thm:AC}There exists a constant $\delta \in (0,1)$ such that for all $\epsilon \in (0,1)$, there is an $\FNP$-machine that, for infinitely many $N$, on input $1^N$ outputs an $N \times N$ matrix that is $(\delta \cdot N^2, 2^{(\log N)^{1/4 - \epsilon}})$-rigid.%
\footnote{Their result is stated for $\FP^{\NP}$ but can be strengthened to $\FNP$. More precisely, on infinitely many inputs $1^N$ that are accepted by the $\FNP$-machine, any accepting path outputs a 
rigid matrix. Matrices obtained on different accepting paths may differ, but all of them are rigid. If one insists on outputting the same matrix on all accepting paths, then this can be done in $\FP^{\NP}$.}
\end{theorem}

Their result still does not attain the rank bounds required for Valiant's lower bounds, yet it vastly improves the state of the art of explicit rigid matrix constructions.
On Razborov's end, the construction indeed meets the required rigidity parameters (in fact, greatly exceeds them), but does not fulfill the requirement of super-explicitness. That said, they use a tensoring argument to obtain $N\times N$ matrices still within Razborov's rigidity parameters, in which each entry is computable in non-deterministic time $2^{(\log\log N)^{\omega(1)}}$. While this is not super-explicit, it is an exponential improvement over previous results.

The surprising construction of Alman and Chen is a tour-de-force that ties together seemingly unrelated areas of complexity theory. A key area is the theory of Probabilistically Checkable Proofs (PCPs). PCPs provide a format of rewriting classical $\NP$-proofs that can be efficiently verified based only on a small amount of random queries into the rewritten proof. The PCP Theorem \cite{AroraS1998,AroraLMSS1998} asserts that any $\NP$-proof can be rewritten into a polynomially-longer PCP that can be verified using a constant number of queries. Alman and Chen make use of \emph{efficient} and \emph{short} PCPs for $\NTIME(2^n)$, as well as \emph{smooth} PCPs. Momentarily, we too will make use of these properties, so let us give an informal description of these:

\begin{description}
\item[Efficient PCP:] A PCP for $\NTIME(T(n))$ is said to be {\em efficient} if the running time of the PCP verifier is sub-linear (or even logarithmic) in the length of the original $\NP$-proof (i.e., $T(n)$).

\item[Short PCP:] A PCP for $\NTIME(T(n))$ is said to be {\em short}  if the length of the PCP is nearly linear (i.e., $T(N)^{1+o(1)}$) in the length of the original $\NP$-proof (i.e., $T(n)$).

\item[Smooth PCP:] A PCP is said to be {\em smooth} if for each input, every proof location is equally likely to be queried by the PCP verifier.
\end{description}

The PCP Theorem has  had a remarkable impact on our understanding of the hardness of approximation of several combinatorial problems (see, e.g., a survey \cite{Trevisan2013}). Parallel to this line of work, Babai, Fortnow, Levin and Szegedy~\cite{BabaiFLS1991} initiated a long sequence of works~\cite{BenSassonSVW2003,BenSassonGHSV2006,BenSassonS2008,BenSassonGHSV2005,Dinur2007,MoshkovitzR2008,Mie2009,BenSassonV2014,Paradise2020} that prove that there exist efficient, short and smooth PCPs for $\NTIME(2^n)$. Alman and Chen's construction makes use of these PCPs, as well as the non-deterministic time-hierarchy theorem \cite{Zak1983} and a fast (i.e., faster than $N^2/\polylog(N)$ time for $N\times N$ matrices) algorithm for counting the number of ones in a low-rank matrix \cite{ChanW2016}.

\subsection{Our results}
Our work arises from asking if there exist PCPs with additional ``nice'' properties that can strengthen the above construction due to Alman and Chen. We answer this question in the affirmative, by (1) introducing a new variant  of PCPs that we refer to as {\em rectangular PCPs}, (2) constructing efficient, short and smooth rectangular PCPs and (3) using these rectangular PCPs to strengthen and simplify the rigid-matrix construction in \cref{thm:AC}. We begin by stating the improved rigid matrix construction.

\newcommand{\maintheoremstatement}{\
There is a constant $\delta \in (0,1)$ such that there is an $\FNP$-machine that for infinitely many $N$, on input $1^N$ outputs an $N \times N$ matrix that is $(\delta \cdot N^2, 2^{\log N/\Omega(\log \log N)})$-rigid.\
}
\begin{theorem}\label{thm:main} \maintheoremstatement
\end{theorem}

We remark that Alman and Chen obtained a conditional result which proved a similar conclusion using the easy witness lemma of Impagliazzo, Kabanets and Wigderson~\cite{ImpagliazzoKW2002}: either $\NQP \not\subset \Ppoly$ or for all $\epsilon \in (0,1)$ there exists an $\FNP$-algorithm that for infinitely many $N$, on input $1^N$ outputs an $N \times N$ matrix that is $(\delta \cdot N^2, 2^{(\log N)^{1-\epsilon}})$-rigid. Our main result (\cref{thm:main}) is thus a common strengthening of both \cref{thm:AC} as well as this conditional result.

\subsection{Rectangular PCP}

Our result is obtained by a new notion of PCPs, called {\em rectangular PCPs}.\footnote{We thank Ramprasad (RP) Saptharishi for suggesting the term ``rectangular PCPs".}
Briefly put, rectangular PCPs are PCPs where the proofs are thought of as square matrices, and the random coins used by the verifier can be partitioned into two disjoint sets, one determining the \emph{row} of each query and the other determining the \emph{column}. To get a better feel for this new property, we examine the constraint satisfaction problem (CSP) underlying a rectangular PCP,\footnote{See, for example, \cite[Chapter 18]{AroraBarak} for a description of the CSP underlying a PCP.} and defer the full definition to \cref{sec:prelim}.

Consider the classical $\NP$-hard  constraint satisfaction problem $\maxcut$ whose instance is a directed graph  $G=(V,E)$ with $n$ vertices and $m$ edges and the goal is to find a subset $S \subseteq V$ that maximizes the number of edges cut  (in either direction) between $S$ and $V\setminus S$. The instance $G$ is {\em rectangular} if the following condition is met.
\begin{itemize}
\item There exist two directed graphs $G_1$ and $G_2$ with $\ell = \sqrt{n}$ vertices and $r = \sqrt{m}$ edges each such that $G$ is the product graph $G_1\times G_2$, i.e., the edges of $G$ satisfy the following {\em rectangular} property:
  \[ ((u_1,u_2) ,(v_1,v_2)) \in E(G) \Longleftrightarrow (u_1,v_1) \in E(G_1) \text{ and } (u_2,v_2) \in E(G_2). \]
\end{itemize}
An instance $G$  is said to be {\em $\tau$-almost rectangular} for $\tau \in [0,1)$ if $G$ is the edge-disjoint union of $m^{\tau}$ product graphs $G^{(j)}_1 \times G^{(j)}_2, j \in [m^\tau]$ where each of the product graphs $G^{(j)}_1 \times G^{(j)}_2$ is defined on the same vertex set $V$ and satisfies $|E(G^{(j)}_1)|= |E(G^{(j)}_2| =  m^{(1-\tau)/2}$. To distinguish rectangular graphs from almost-rectangular graphs, we will sometimes refer to them as perfectly rectangular.

This definition of rectangularity can be extended to arbitrary $q$-CSPs as follows. Let $\Phi$ be a $q$-CSP instance on a set $V$ of $n=\ell^2$ variables. Let $\calC$ be the set of $m=r^2$ constraints of $\Phi$. As both the number of variables ($n=\ell^2$) and the number of constraints ($m=r^2$) are perfect squares, we will w.l.o.g. index them with double indices, $(i_i,i_2) \in [\ell] \times [\ell]$ and $(j_1,j_2) \in [r]\times [r]$. Let the $(j_1,j_2)$-th constraint in $\calC$ involve the $q$ variables $x_{c_1(j_1,j_2)},\dots, x_{c_q(j_1,j_2)}$. The instance $\Phi$ is said to be rectangular if for any $k\in[q]$, the address function $c_k: [r]\times [r] \to [\ell] \times [\ell]$ that specifies the $k$-th variable in the $(j_1,j_2)$-th clause can be decomposed into a product function $a_k \times b_k$ where $a_k, b_k \colon [r]\to [\ell]$. Almost rectangularity is defined similarly.

Back in the ``proof systems'' view, a PCP is said to be {\em ($\tau$-almost) rectangular} if its underlying CSP is $(\tau$-almost) rectangular.
Thus, rectangularity can be viewed as natural structural property referring to the clause-variable relationship in the CSPs produced by the PCP. Our main technical result is that there exists an efficient, short, smooth and almost-rectangular PCP. 
A simplified  version of our result is as follows (see \cref{thm:combinedPCP} for the exact statement.)

\begin{theorem}\label{thm:recPCP}
Let $L$ be a language in $\NTIME(2^n)$. For every constants $s \in (0,1/2)$ and $\tau \in (0,1)$, there exists a constant-query, smooth and $\tau$-almost rectangular PCP for $L$ over the Boolean alphabet with perfect completeness, soundness error $s$, proof length at most $2^n \cdot\poly(n)$ and verifier running time at most $2^{O(\tau n)}$.
\end{theorem}

\subsection{Rectangular PCPs to rigid matrices}\label{sec:pcps_to_rigid_intro}

We now sketch how rectangular PCPs can be used to construct rigid matrices. This will also serve as a motivation for the definition of rectangular PCPs. Our construction follows that of Alman and Chen (which fits within the \emph{lower bounds from algorithms} framework of Williams~\cite{Williams2013} -- see Section~\ref{sec:related}), with the main difference being the use of rectangular PCPs. We show that this simplifies their construction and improves the rigidity parameters it attains. The construction is inspired by the subtitular maxim: 
\begin{quote}
  {\em Hard claims have complex proofs.}
\end{quote}
Informally speaking, the construction is an instantiation of this maxim where ``complexity'' refers to rigidity, ``proofs'' are PCPs, and ``hard claims'' are instances of the hard language guaranteed by the non-deterministic time hierarchy theorem (see next).

The main ingredients in our construction are as follows:
\begin{enumerate}
\item The non-deterministic time-hierarchy theorem~\cite{Zak1983}: There exists a unary language $L \in \NTIME(2^n) \setminus \NTIME(2^n/n)$.
\item A non-trivial (i.e., sub-quadratic time) algorithm to compute the number of ones in a low-rank $\{0,1\}$-valued matrix when given as input its low-rank decomposition $N = P \cdot Q$, where $P$ and $Q$ are matrices of dimensions $N\times \rho$ and $\rho\times N$, respectively. Such results with running time $N^{2-\epsilon(\rho)}$ were developed by Chan and Williams~\cite{ChanW2016} with $\varepsilon\pars{\rho} = \Omega\pars{1/\log \rho}$.
\item The existence of efficient, short, smooth and rectangular PCPs for $\NTIME(2^n)$, as guaranteed by \cref{thm:recPCP}.
\end{enumerate}

  Let $L \in \NTIME(2^n) \setminus \NTIME(2^n/n)$ as guaranteed by the non-deterministic time-hierarchy theorem. By \cref{thm:recPCP}, there exist efficient, short and smooth rectangular PCPs for $L$. Our goal is to show that either (a natural transformation of) the PCP yields a rigid matrix, or $L \in \NTIME(2^n/n)$ -- a contradiction.
  
  For simplicity of presentation, we will assume that the rectangular PCPs obtained in \cref{thm:recPCP} are perfectly rectangular and furthermore that the underlying CSP of the PCP is $\maxcut$ with completeness $c$ and soundness $s$ for some constants $0<s<c<1$. In other words, the PCP reduction reduces instances $x \in L$ to digraphs $G$ which have a fractional cut of size at least $c$, and instances $x \notin L$ to digraphs $G$ which do not have any cut of fractional size larger than $s$.
  
  Let us understand what it means for the PCP to be short, smooth, efficient and rectangular: \emph{``Rectangular''} refers to the fact that the digraph $G$ is a product graph $G_{1} \times G_{2}$; \emph{``Short''} implies that the size of $G$ (i.e., the number of vertices and edges) is at most $r^2=2^n\cdot \poly(n)$; \emph{``Smooth''} implies that the digraph $G$ is regular (the degree of a vertex is the sum of its in-degree and out-degree); and \emph{``Efficient and rectangular''} implies that for each of the two graphs $G_{1}$ and $G_{2}$, given an edge the vertices incident on the edge can be obtained in time $2^{\gamma n}$ (for a small constant $\gamma>0$ of our choice). 
  
  For any instance $x$ of the language $L$, any cut of the corresponding graph $G$ is of the form $(S, V(G) \setminus S)$. Since $V(G) = V(G_{1}) \times V(G_{2})$, we can identify the cut $S$ with a $V(G_{1})\times V(G_{2})$-matrix with $\{0,1\}$ entries. Let $L_1, R_1 \in \{0,1\}^{E(G_{1}) \times V(G_{1})}$ be the incidence matrices indicating the left and right endpoints of the edges in $G_{1}$ (i.e., if $e=(u,v) \in E(G_{1})$ then $L_1(e,u) = R_1(e,v) = 1$). Similarly, define matrices $L_2, R_2 \in \{0,1\}^{E(G_{2}) \times V(G_{2})}$. These matrices will not be computed explicitly; efficiency of the PCP implies that for any given row-index, the non-zero column of that row can be computed in time $2^{\gamma n}$. We will refer to this fact as the \emph{somewhat-efficient} computation of these matrices.
  
 Observe that the matrix $L_1\cdot S \cdot L_2^T$ is an indicator matrix indicating if the left endpoint of the edge is in the set $S$ or not. Similarly $R_1\cdot S \cdot R_2^T$ refers to the indicator of the right endpoint of the edge. Hence, the matrix $M(S):= L_1\cdot S \cdot L_2^T + R_1\cdot S \cdot R_2^T$ is the indicator matrix of whether the edge is cut by the set $S$ or not (with addition over $\F_2$). Thus, the size of the cut induced by the set $S$ is exactly the number of ones in the matrix $M(S)$. Let us denote this quantity by $\val(S):= \numones(M(S))$.

We will prove that for infinitely many $x \in L$, every cut $S^* \in \{0,1\}^{V(G_{1})\times V(G_{2})}$ of fractional size at least $c$ is a $(\delta \cdot N^2 , \rho)$-rigid matrix, for $\delta = (c-s)/3$ and $\rho = 2^{n/\Omega(\log n)}$. Assume towards contradiction that this was not the case. Then, there for long enough $x \in L$, there exists a cut $S^\ast$ that is non-rigid. Since $S^*$ is non-rigid, it is $\delta$-close to some Boolean matrix $S=P \cdot Q$ such that $P$ and $Q$ are Boolean matrices of dimensions $V(G_{1})\times \rho$ and $\rho\times V(G_{2})$, respectively. We now make two observations.

  \begin{itemize}
  \item Since the cut $S^*$ is of size at least $c$ and is $\delta$-close to $S$, it follows from the regularity of $G$ that the cut induced by the set $S$ is of size at least $c-2\delta$. In other words, $\val(S) \geq c -2\delta$.
  \item We can compute $\val(S) = \numones(M(S))$ in time $ O\left( r\cdot (\rho + 2^{\gamma n}) + r^{2-\epsilon(2\rho)}\right) $ as follows. Recall that $M(S)= L_1\cdot S \cdot L_2^T + R_1\cdot S \cdot R_2^T$ and $S = P \cdot Q$. Hence,
  \begin{align*}
      M(S) & =  L_1\cdot P \cdot Q \cdot L_2^T + R_1\cdot P \cdot Q \cdot R_2^T = \underbrace{\begin{pmatrix} L_1 & R_1 \end{pmatrix} \cdot \begin{pmatrix} P & 0 \\ 0 & P\end{pmatrix}}_{=:\widetilde{P}} \cdot \underbrace{\begin{pmatrix} Q & 0 \\ 0 & Q\end{pmatrix} \cdot \begin{pmatrix} L_2^T \\ R_2^T \end{pmatrix}}_{=:\widetilde{Q}} 
  \end{align*}
  So $M(S)$ is a matrix of rank at most $2\rho$ with a low-rank decomposition given by $M(S) = \widetilde{P}\cdot \widetilde{Q}$. Given matrices $P$ and $Q$ and the somewhat-efficient computation of the matrices $L_1, L_2, R_1, R_2$, the matrices $\widetilde{P}$ and $\widetilde{Q}$ may be computed in time $r\cdot (\rho +  2^{\gamma n})$. Finally, we invoke the algorithm of Chan and Williams to compute $\val(S) = \numones(\widetilde{P}\cdot \widetilde{Q})$ in time $O\pars{r^{2 - \epsilon (2\rho)}}$.
  \end{itemize}
This suggests the following non-deterministic algorithm for checking membership in $L$
\begin{itemize}
  \item On input $1^n$
    \begin{enumerate}
    \item Non-deterministically guess matrices $P \in \{0,1\}^{\ell \times \rho}$ and $Q \in \{0,1\}^{\rho\times \ell}$.
    \item Use the efficient PCP verifier to somewhat-efficiently compute the matrices $L_1, L_2, R_1, R_2$.
    \item Compute the matrices $\widetilde{P}$ and $\widetilde{Q}$.
    \item Compute the number of ones $\nu$ of the matrix $\widetilde{P}\cdot \widetilde{Q}$.
    \item Accept if and only if $\nu  > s \cdot r^2$.
\end{enumerate}
\end{itemize}
Indeed, this algorithm decides $L$: If $1^n \in L$, then there exists a guess $S = P \cdot Q$ that would get $\val(S) \ge (c-2\delta)\cdot r^2 > s \cdot r^2$. On the other hand, if $1^n \notin L$, then by the soundness of the PCP, any cut $S$ would have $\val(S) \le s\cdot r^2$.

However, for a suitable choice of $\rho$, this algorithm runs in time $O\left(\ell\cdot \rho + r\cdot (\rho + 2^{\gamma n}) + r^{2-\epsilon(2\rho)}\right) = O(2^n/n)$ -- contradicting the time-hierarchy theorem. Hence, it is false that for every long enough $x \in L$, there exists a cut $S^* \in \{0,1\}^{V(G_{1})\times V(G_{2})}$ of fractional size at least $c$ which is a non-rigid matrix. This immediately yields an $\FNP$-algorithm that infinitely often outputs rigid matrices.

\medskip
In \cref{sec:pcps_to_rigid} we complete this sketch into a full proof that deals with two significant caveats: the PCP is only almost-rectangular, and the predicate of the PCP is not necessarily $\maxcut$. The first is not a significant obstacle and the generalization is rather immediate. The second requires more care, but examining the proof reveals that we only used the fact that each clause has the same predicate or, in PCP jargon, that the predicate is \emph{oblivious to the randomness}. To this end, we define a property of PCPs wherein the predicates have randomness-oblivious-predicates (ROP) and show that the rectangular PCPs constructed in \cref{thm:recPCP} can also be made ROP (see \cref{sec:rop,section:all_together} for exact details). 

\paragraph{Comparison with the Alman--Chen construction:} Alman and Chen obtained a similar result conditioned on the assumption $\NQP \subseteq \P/\poly$, using the easy witness lemma. To obtain an unconditional result, they used a bootstrapping argument 
which results in rigidity for rank at most $2^{(\log N)^{1/4-\epsilon}}$. 
The above proof, on the other hand, is not conditioned on any assumption, does not require the easy witness lemma, and implies rigid matrices for rank $2^{\log N/\Omega(\log \log N)}$. In fact, there is almost no loss due to the PCPs in the above argument. For instance, if the number of ones in $N \times N$ matrices of rank $N^{0.999}$ could be computed in sub-quadratic time, then our construction would yield matrices rigid for rank $N^{0.99}$.

\subsection{Constructing rectangular PCPs} \label{sec:outline_intro}
  
The rectangular property of PCPs states that the underlying constraint satisfaction problem (CSP) has a product structure.

\subsubsection*{Warm-up: Rectangularity of some known constructions}
As a warm-up, let us examine the rectangularity of some common PCP building blocks. The purpose of this warm-up is to become comfortable with the notion of rectangularity. Towards this end, we chose some simple examples from the PCP literature, rather than examples that are actually used in our construction.

First, it is immediate that PCPs obtained from \emph{parallel repetition} are rectangular. Unfortunately, the size of these PCPs are far from being nearly-linear.

Next, recall the Blum--Luby--Rubinfeld (BLR) \emph{linearity tester} \cite{BlumLR1993} that checks if a given function $f: \F_2^m \to \F_2$ is linear.

\begin{algorithmNoNum}[BLR Tester]
    On oracle access to $f \colon \F_2^m \to \F_2$,
    \begin{enumerate}
    \item Sample $x, y \in_R \F_2^m$.
    \item Query $f$ at locations $x$, $y$, and $x+y$.
    \item Accept if and only if $f(x) +f(y)+f(x+y)=0$.
    \end{enumerate}
\end{algorithmNoNum}
The $(x,y)$-th constraint in the above test queries the three locations $x,y,x+y\in \F_2^{3m}$. For even $m$, we can write $x = (x_1,x_2)$ and $y=(y_1,y_2)$ where $x_1,x_2,y_1,y_2 \in \F_2^{m/2}$. Thus, the $((x_1,x_2),(y_1,y_2))$-test queries the three locations $(x_1,x_2)$,  $(y_1,y_2)$ and $(x_1+y_1,x_2+y_2)$. Hence, the BLR test is perfectly rectangular. For similar reasons, the low-degree test (actually used in our construction) is also perfectly rectangular.

\subsubsection*{The actual construct}

The warm-up gives us hope that PCPs constructed using from low-degree test are rectangular or can be made so with some modification. Our construction, essentially, realizes this hope. In particular we take a closer look at the short and efficient PCP construction of Ben-Sasson~\etal~\cite{BenSassonGHSV2006,BenSassonGHSV2005} and modify it suitably to obtain a rectangular PCP. This is a rather delicate operation and involves several subtleties along the way. We highlight the salient steps in the construction below.

For starters, recall another key ingredient in the construction of PCPs: the composition paradigm of Arora and Safra \cite{AroraS1998}. We will use the modular composition paradigm of Ben-Sasson~\etal~\cite{BenSassonGHSV2006} and Dinur and Reingold~\cite{DinurR2006}, wherein a \emph{robust} PCP is composed with a PCP \emph{of proximity}. Our construction of rectangular PCPs will proceed along the following lines.
\begin{enumerate}
\item \label{intro_outline_bghsv} We first show that the Reed--Muller based PCP construction due to Ben-Sasson~\etal~\cite{BenSassonGHSV2006,BenSassonGHSV2005} can be modified to yield a short almost-rectangular robust PCP. This involves a careful, step-by-step examination of this PCP. As indicated above, the low-degree component of this PCP is perfectly rectangular. However, this PCP also involves a sum-check component, which is inherently \emph{not} rectangular, but is fortunately \emph{almost} rectangular.
\item We then show that composition of an almost-rectangular robust PCP with a (not necessarily rectangular) PCP of proximity  yields an almost-rectangular PCP.
\end{enumerate}
Composing the outer robust PCP obtained in \cref{intro_outline_bghsv} with the short and efficient PCP of proximity of Mie~\cite{Mie2009} yields a short, efficient and rectangular PCP with constant query complexity. However, this PCP is not necessarily smooth. By now, there are several standard techniques to ``smoothify'' a PCP in literature, but these techniques do not necessarily retain the rectangular property. To obtain a rectangular and smooth PCP, we actually work with a stronger notion of rectangularity, that we refer to as ``rectangular-neighborhood-listing (RNL)'' and show that a short and efficient PCP with RNL can be ``smoothified'' to yield the desired short, efficient, smooth and rectangular PCP.

\subsection{Related Work}\label{sec:related}

\paragraph{PCPs with structured queries.}
We view \cref{thm:recPCP} as continuing a line of work that explores the connection between the randomness of a PCP and the structure of its queries. A prominent advance in this direction is the work of Ben-Sasson and Viola~\cite{BenSassonV2014}. They constructed short and efficient PCPs in which queries are a function of the input and a simple projection of the randomness (namely, a \emph{$1$-local} function: for a fixed input, each bit in each query location is a fixed bit of the randomness or its negation).%
\footnote{Interestingly, the construction of \cite{BenSassonV2014} is also an adaptation of a PCP from \cite{BenSassonGHSV2005}, but not the same one as in our work. Namely, we build upon the Reed--Muller based PCP in \cite{BenSassonGHSV2005}, whereas \cite{BenSassonV2014} builds upon the Reed--Solomon based PCP in the same paper.}
Although the query structure in \cite{BenSassonV2014} (and follow-up~\cite{Viola2020}) is very simple, it is unclear whether their PCP is almost-rectangular or smooth---both playing a crucial role in our construction and its application.

In a different direction, Feige and Jozeph \cite{FeigeJ2012} constructed PCPs in which the queries depend only on the randomness but not on the input. Recently, Austrin, Brown-Cohen, and H{\aa}stad~\cite{AustrinBH2019} improved this result to have optimal soundness error for certain verification predicates such as $\threeSAT$ and $\threeLIN$.

\paragraph{Circuit Lower Bounds from Algorithms.}
The maxim ``hard claims having complex proofs'' is inspired by a result of Williams~\cite{Williams2016}, showing that witnesses for $\NTIME(2^n)\setminus \NTIME(2^n/n)$ cannot be truth-tables of certain small-size low-depth circuits (specifically, $\ACC^0$ circuits). That work is a part of Williams's algorithmic approach to circuit lower bounds originating in \cite{Williams2013, Williams2014}.
Roughly speaking, Williams's framework shows how to obtain lower bounds against a certain circuit class by designing non-trivial (i.e., better than exhaustive search) SAT algorithms for circuits in the class.
Williams~\cite{Williams2013} also observed the usefulness of PCPs within this framework: using PCPs, one can obtain circuit lower bounds from any non-trivial derandomization.%
\footnote{More precisely, from any non-trivial deterministic estimation of the acceptance probability of a circuit, up to a constant additive error.}
Santhanam and Williams~\cite{SanthanamW2014},  Ben-Sasson and Viola~\cite{BenSassonV2014}, and Chen and Williams~\cite{ChenW2019} further explored and tightened this connection.
In this light, the overall proof strategy of Alman and Chen~\cite{AlmanC2019} can be seen as a surprising instantiation of Williams's framework for average-case hardness of the computational model of low-rank matrices.

\paragraph{Applications to Probabilistic Degree.} Recently (and independently of this work), Viola~\cite{Viola2020} showed the existence of functions on $n$ variables in $\E^\NP$ with approximate probabilistic degree $\Omega(n/\log^2 n)$ over $\F_2$, for infinitely many $n \in \N$. Using the known relation between matrix rigidity and approximate rank (see \cite[Proposition 7.5]{AlmanC2019}), \cref{thm:main} implies a similar lower bound on the approximate probabilistic degree.



\subsection{Organization}

The rest of the paper is organized as follows. 
\begin{description}
\item[Preliminaries (\cref{sec:prelim}).] We begin by giving a definitional treatment of PCPs and its variants. In particular, we formally define the rectangular PCP, which is the central object of our focus. We also define the two aforementioned related properties: rectangular-neighborhood-listing (RNL) and randomness-oblivious-predicates (ROP).
\item[From Rectangular PCPs to Rigid Matrices (\cref{sec:pcps_to_rigid}).] We show how the existence of efficient, short and smooth rectangular PCPs with ROP for $\NTIME(2^n)$ yields rigid matrices (thus proving \cref{thm:main}, modulus the actual rectangular PCP construction).
\item[A Construction of Rectangular PCPs (\cref{sec:smoothification,section:outer,sec:rop,section:composition,section:all_together}).] In the remaining sections of the paper, we construct efficient, short and smooth rectangular PCPs for $\NTIME(T(n))$. The main steps in the construction are as follows:
\begin{itemize}
\item \cref{sec:smoothification}: We show how any PCP with RNL and ROP can be converted to a smooth and rectangular PCP with ROP. Hence, from this point onwards, we seek PCPs with RNL, rather than rectangular PCPs.
\item \cref{section:outer}: We show that the robust PCP verifier of Ben-Sasson~\etal~\cite{BenSassonGHSV2006,BenSassonGHSV2005} has RNL. 
\item \cref{sec:rop}: We show how to add ROP to any robust PCP with RNL.
\item \cref{section:composition}: We then show that any PCP of proximity, when composed with a robust PCP that has RNL and ROP, yields a PCP with RNL and ROP. Note that the PCP of proximity need not be rectangular.
\item \cref{section:all_together}: Finally, we combine the results proved in \cref{sec:smoothification,section:outer,sec:rop,section:composition}  to obtain our main construct: an efficient, short and smooth rectangular PCP (thus proving \cref{thm:recPCP}).
\end{itemize}
\end{description}


\section{PCPs: definitions and variants} 
\label{sec:prelim}

The main focus of this section is to introduce the notion of {\em rectangular PCPs}, the central object of interest in this work. To this end, we begin by recalling the standard definition of PCPs and related objects (PCP verifier, robust soundness, smooth PCPs) before proceeding to define rectangular PCPs. 

\paragraph{Notation.} Let $\Sigma$ be any finite alphabet. 
For $u,v\in\Sigma^n$, the relative Hamming distance between $u$ and $v$, denoted by $\delta(u,v)$, is the fraction
of locations on which $u$ and $v$ differ
(i.e., $\delta(u,v)\coloneqq|\{i:u_i\neq v_i\}|/n$).
We say that $u$ is {\em $\delta$-close} to $v$ (resp., {\em $\delta$-far} from $v$)
if $\delta(u,v)\leq \delta$ (resp., $\delta(u,v)>\delta$).
The relative distance of a string $u$ to a set $V$ of strings
is defined as $\delta(u,V)\coloneqq\min_{v\in V}\{\delta(u,v)\}$.

\subsection{Standard PCPs} 

We begin by recalling the formalism of an efficient PCP verifier. As is standard in this literature, we restrict our attention to {\em non-adaptive} verifiers. 

\begin{definition}[efficient PCP verifiers]
\ \\\begin{itemize}
\vspace{-0.5cm}\item
Let $r,q,m,d,t, \sigma \colon \pintegers\to\pintegers$. A \defemph{$(r,q,m,d,t)$-restricted verifier} over alphabet $\Sigma \coloneqq \bool^\sigma$ is a probabilistic
algorithm $V$ that, on an input $x$ of length $n$,
tosses $r \coloneqq r(n)$ random coins $R$
and generates a sequence of $q \coloneqq q(n)$ \defemph{query locations} $I \coloneqq (i^{(1)},\ldots,i^{(q)})$, where each $i^{(k)} \in [m(n)]$, 
and a \defemph{(decision) predicate} $D \colon \Sigma^q\rightarrow \{0,1\}$ of size at most $d(n)$ in time at most $t(n)$.

Think of $V$ as representing a probabilistic oracle machine that
queries the proof oracle $\pi\in \Sigma^m$, for the positions in $I$, receives the $q$
symbols $\pi\restrict{I} \coloneqq (\pi_{i^{(1)}},\ldots,\pi_{i^{(q)}})$,
and accepts iff $D(\pi\restrict{I})=1$.
\item
We write $(I,D)\getsr  V(x)$
to denote the queries and predicate generated by $V$ on input $x$ and
random coin tosses. To explicitly mention the random coins $R$, we write $(I,D) \gets V(x;R)$.
\item
We call $r$ the \defemph{randomness complexity}, $q$ the \defemph{query complexity}, $m$ the \defemph{proof length}, $d$ the \defemph{decision complexity} and $t$ the \defemph{running time} of $V$. $\sigma$ is called the \emph{answer complexity} of $V$, and will usually be omitted.\footnote{All PCPs in this work will be Boolean (i.e., $\sigma = 1$), except for an intermediate PCP in \cref{section:outer}. Even there, it will be more convenient to consider the alphabet size $\abs{\Sigma} = 2^\sigma$ rather than $\sigma$.}
\end{itemize}
\end{definition}

It will be convenient at times to have the following graphical description of the verifier. Given a $(r,q,m,d,t)$-restricted verifier and input $x$, consider the bipartite graph $G(V,x)\coloneqq(L = \{0,1\}^r, R=[m], E)$ where $(R, i) \in E$ if the verifier $V$ on input $x$ and random coins $R$ queries location $i$ in the proof. Clearly,  the graph $G(V,x)$ is $q$-left regular.

We can now define the standard notion of efficient PCPs with perfect completeness. 

\begin{definition}[PCP]\label{def:PCP}
For a function $s : \pintegers\to [0,1]$, a verifier $V$ is
a \defemph{probabilistically checkable proof system (PCP)} for a language $L$
with \defemph{soundness error} $s$ if the following two conditions hold
for every string $x$:
\begin{description}
\item[Completeness:] If $x \in L$ then there exists $\pi$
such that $V(x)$ accepts oracle $\pi$ with probability~1.
Formally, \[\exists \pi \qquad \Pr_{(I,D)\getsr V(x)}[D(\pi\restrict{I}) = 1]=1.\]
\item[Soundness:] If $x \not \in L$ then
for every oracle $\pi$, the verifier $V(x)$ accepts $\pi$ with probability
strictly less than~$s$. Formally, \[\forall \pi \qquad
\Pr_{(I,D)\getsr V(x)}[D(\pi\restrict{I}) = 1] < s(|x|).\]
\end{description}
\end{definition}

While constructing PCPs, we will sometimes be interested in PCPs with a stronger notion of soundness, referred to as {\em robust soundness}. 

\begin{definition}[robust soundness]
For functions
$s,\rho : \pintegers\rightarrow [0,1]$,
a PCP verifier $V$ for a language $L$
has \defemph{robust-soundness error} $s$ with \defemph{robustness parameter $\rho$}
if the following holds for every $x\notin L$:
For every oracle $\pi$,
the symbols read by the verifier $V$ are $\rho$-close to being
accepted with probability
strictly less than $s$. Formally,
\[\forall \pi
\Pr_{(I,D)\getsr V(x)}
    [\mbox{$\exists a$ s.t. $D(a)=1$ and $\delta(a,\pi\restrict{I})\leq \rho$}]
   < s(|x|).\]
\end{definition}

 By now, we know of several such efficient PCP constructions, one of which we state below.

\begin{theorem}[efficient PCPs for $\NTIME(T)$~{\cite[Theorem~2.6]{BenSassonGHSV2005}}]\label{thm:bghsv2}
Suppose that $L$ is a language in
$\NTIME(T(n))$ for some non-decreasing function $T \colon \N \rightarrow \N$. Then for every $\epsilon \in (0,1)$, $L$ has a PCP verifier over $\bool$ with soundness error $\epsilon$, query complexity $O(1/\epsilon)$ and randomness complexity $\log T(n) +\log^{O(\epsilon)} T(n)$.
\end{theorem}

While constructing variants of the above PCP, we will particularly be interested in smooth PCPs. 

\begin{definition}[smooth PCP]\label{def:smooth_PCP}
Given a $(r,q,m,d,t)$-restricted verifier $V$, an input $x$ and $i \in [m]$, 
let $Q_x(i)$ denote the probability with which the verifier $V$ outputs $i$ on a random query $k\in [q]$. Formally, \[ Q_x(i) \coloneqq \Pr_{R, k \in [q]} [ i^{(k)} = i | (I,D) \gets V(x;R) ].\]
The PCP verifier $V$ is said to be smooth if for all $i, j \in [m]$, $Q_x(i) = Q_x(j)$. 
\end{definition}

Thus, smooth PCPs refer to PCPs whose verifiers query all locations of the proof oracle equally likely (or equivalently in the above graphical description,  verifiers whose corresponding bipartite graphs are also right-regular).\footnote{\emph{Minor historical inconsistencies in the definition of smoothness:} Several previous works \cite{KatzT2000,Paradise2020} defined a smooth oracle machine as one in which each location of the oracle has equal probability of being queried by the machine \emph{in any of its queries} (rather than in a random query, as in \cref{def:smooth_PCP} as well as other prior works \cite{GoldreichS2006,BenSassonGHSV2006}). Indeed, both definitions are equivalent assuming the machine never queries the same location twice for any given random coin sequence $R$. Our definition is more convenient as it coincides with right-regularity of the corresponding bipartite graph even without this assumption.}

\begin{remark}[tolerance of smooth PCPs]
	A smooth PCP is \defemph{tolerant} of errors in a correct proof, in the sense that a proof that is close to a correct one is accepted with good probability. Concretely, suppose $V$ makes $q$ queries to its proof and is smooth. Then if $\pi$ is a correct proof for $V$ (i.e. accepted w.p. $1$) and $\pi^\ast$ is $\delta$-close to $\pi$ in relative Hamming distance, then $\pi^\ast$ is accepted with probability at least $1- q \cdot \delta$.
\end{remark}

The PCPs constructed in Theorem~\ref{thm:bghsv2} are not necessarily smooth, however they can be made smooth without too much of an overhead. In this work we will be interested in {\em smoothening} the PCP maintaining yet another property, {\em rectangularity}, which we introduce in the following section.

\subsection{Rectangular PCPs}

We now define {\em rectangular PCPs}, the central object of interest in this work. As the name suggests, rectangular PCPs are PCPs in which the proof oracle $\pi:[m]\to \Sigma$, an $m$-length string,  is interpreted as a matrix $\pi:[\ell]\times [\ell] \to \Sigma$ for some $\ell$ such that $m = \ell^2$ (yes, we assume that the proof lengths are always squares of integers). Furthermore, the verifier is also ``rectangular'' in the sense that the randomness $R \in \{0,1\}^r$ is also partitioned into 2 parts $R = (R_\row, R_\column)$ such that the row index of the queries is obtained from the ``row randomness'' $R_\row$ while the column index of the queries is obtained from the ``column randomness'' $R_\column$.

The above informal description assumes ``perfect'' rectangularity while the definition below allows for the relaxed notion of ``almost-rectangularity'', in which randomness is partitioned into three parts: row and column (as above), as well as a small \emph{shared} part that is used for obtaining both the rows and the columns of the queries.

\begin{definition}[Rectangular PCP]\label{def:rectangular}
For $\tau \in [0,1)$, a $(r,q,\ell^2,d,t)$-restricted verifier $V$ is said to be \defemph{$\tau$-rectangular} if the following holds.

The random coin tosses $R \in \{0,1\}^r$ can be partitioned into 3 parts

\[ R = (R_\row, R_\column, R_\shared)  \in \{0,1\}^{(1-\tau)r/2} \times \{0,1\}^{(1-\tau)r/2} \times \{0,1\}^{\tau r},\] 
such that the verifier $V$ on input $x$ of length $n$ and random coins $R$ produces a sequence of $q$ \defemph{proof locations} $I = ((i_\row^{(1)},i_\column^{(1)}),\ldots, (i_\row^{(q)},i_\column^{(q)}))$ as follows:
\begin{itemize}
    \item $I_\row \coloneqq (i_\row^{(1)},\ldots, i_\row^{(q)}) = V_\row(x; R_\row, R_\shared)$,
    \item $I_\column \coloneqq (i_\column^{(1)},\ldots, i_\column^{(q)}) = V_\column(x; R_\column, R_\shared)$,
    \item Generating $I_\row$, $I_\column$, and the decision predicate\footnote{It is natural to wonder how the decision predicate depends on the randomness. This is considered in \cref{sec:rop_def}.} take a total of at most $t(n)$ time.
\end{itemize}
\end{definition}
In other words, the row (respectively column) indices of the queries are only a function of the row (respectively column) and shared parts of the randomness. If $\tau = 0$, we will say the verifier $V$ is \defemph{perfectly rectangular}, and otherwise $V$ is \defemph{almost rectangular}.
When it is obvious from context, we will say that $V$ is simply \defemph{rectangular}, omitting the ``$\tau$-'' qualifier.

\subsection{Rectangular Neighbor-Listing (RNL)}

A careful reading of the construction of PCPs mentioned in Theorem~\ref{thm:bghsv2} will reveal that they are in fact rectangular. However, for our application, we will need rectangular PCPs that are also smooth. Later, we will ``smoothen'' a PCP while maintaining its rectangularity (see \cref{sec:smoothification}), for which we need a stronger property that we refer to as {\em rectangular-neighbor-listing (RNL)}. To define this property, we first define {\em configurations} and {\em neighboring configurations}.

\begin{definition}[configurations and neighboring configurations\label{def:neighbor}]
Given a $(r,q,m,d,t)$-restricted verifier $V$ and an input $x$, a configuration refers to a tuple $(R,k) \in \{0,1\}^r \times [q]$ composed of the randomness of the verifier and query index. The verifier $V$ describes how to obtain the query location $i^{(k)} \in [m]$ from the configuration $(R,k)$ (and the input $x$). 

We say that two configurations $(R,k)$ and $(R',k')$ of a PCP verifier $V$ on input $x$
are \defemph{neighbors} if they both yield the same query location $i \in [m]$. (In particular, every configuration is a neighbor of itself.)
\end{definition}

In the graphical representation of a verifier, a configuration refers to an edge of the bipartite graph and two configurations are said to be neighbors if they are incident on the same right vertex. A configuration $(R,k) = (R_\row,R_\column,R_\shared,k)$ of a rectangular PCP can be broken down into a \defemph{row configuration} $(R_\row,R_\shared,k)$ and a \defemph{column configuration} $(R_\column,R_\shared,k)$. Rectangularity states that the query location $(i_\row^{(k)}, i_\column^{(k)}) \in [\ell] \times [\ell]$ satisfy that $i_\row^{(k)}$ is a function of the row configuration while $i_\column^{(k)}$ is a function of the column configuration. 

\begin{definition}[rectangular neighbor-listing (RNL)]\label{def:RNL}
For $\tau \in [0,1)$ and $t_\RNL : \N \rightarrow \N$, an $(r,q,m,d,t)$-restricted verifier $V$ is said to have the \defemph{$\tau$-rectangular neighbor listing property ($\tau$-RNL)} with time $t_\RNL\pars{n}$ if the following holds.

\begin{itemize}
\item The random coin tosses $R \in \{0,1\}^r$ can be partitioned into 4 parts 
\[ R = (R_\row, R_\column, R_\sharedrow,R_\sharedcolumn)  \in \{0,1\}^{(1-\tau)r/2} \times \{0,1\}^{(1-\tau)r/2} \times \{0,1\}^{\tau r/2} \times \{0,1\}^{\tau r/2},\] 
where we refer to the 4 parts $R_\row, R_\column, R_\sharedrow,R_\sharedcolumn$ as the \defemph{row part}, \defemph{column part}, \defemph{row-shared part} and \defemph{column-shared part} respectively. We will refer to the combined shared randomness $R_\shared\coloneqq(R_\sharedrow,R_\sharedcolumn)$ as the \defemph{shared part}. 

\item There exist two algorithms, a \defemph{row agent} (denoted $A_\row$) and a \defemph{column agent} (denoted $A_\column$) that list, in time $t_\RNL(n)$, all neighbors of a given configuration $\pars{R,k}$ in the following ``rectangular and synchronized'' fashion:
\begin{itemize}
    \item On input a row configuration $\pars{R_\row, R_\shared, k}$, the row agent $A_{\row}$ outputs a list $L_\row$ 
    of tuples $\pars{R^\prime_\row, R^\prime_\sharedrow, k^\prime}$.
    \item On input a column configuration $\pars{R_\column, R_\shared, k}$, the column agent $A_{\column}$ outputs a list $L_\column$ 
    of tuples $\pars{R^\prime_\column, R^\prime_\sharedcolumn, k^\prime}$.
\end{itemize}
satisfying the following properties
\begin{enumerate}
\item \label{rnl_zipped}The two lists $L_\row$ and $L_\column$ are of equal length and entrywise-matching $k^\prime$ values, such that the ``zipped'' list
\begin{equation}\label{List}
		L \coloneqq \left\lbrace
		\pars{R^\prime_\row, R^\prime_\column, R^\prime_\sharedrow, R^\prime_\sharedcolumn, k^\prime}
		\;\;\middle|\;
		\begin{array}{c}
		    i \in \brackets{\abs{L_\row}} \\
			\pars{R^\prime_\row, R^\prime_{\sharedrow}, k^\prime} \coloneqq L_\row \brackets{i} \\
			\pars{R^\prime_\column, R^\prime_{\sharedcolumn}, k^\prime} \coloneqq L_\column \brackets{i}
		\end{array}
		\right\rbrace
\end{equation}
is the list of all full configurations that are neighbors of $\pars{R, k}$.

\item \label{rnl_sync} Not only are the contents of $L$ the same for each two neighboring locations, but the \emph{order} of configurations in $L$ is the same too. That is, for any two neighboring configurations $\pars{R,k}$ and $\pars{\widetilde{R}, \widetilde{k}}$, the resulting configuration lists $L$ and $\widetilde{L}$ are equal as ordered lists (element-by-element).

\item \label{rnl_index} Both agents know the index of $\pars{R,k}$ in the list $L$ (despite not knowing $\pars{R,k}$ entirely).

\end{enumerate}
\end{itemize}
\end{definition}

Informally speaking, rectangularity asserts that the query location can be obtained in a ``rectangular'' fashion from the randomness, while RNL asserts that the entire list of neighboring configurations of the query location can be obtained in a ``rectangular'' fashion. 

A PCP with RNL can be made \emph{smooth} and \emph{rectangular}, as shown in \Cref{sec:smoothification}.

\begin{remark}
\label{remark:BarakGoldreich}
    Barak and Goldreich \cite{BarakG2008} defined PCPs with a \emph{reverse-sampling} procedure that outputs a uniformly random neighbor of any given configuration. The important difference between RNL and reverse-sampling is that the former offers a procedure that outputs neighboring configurations in a rectangular fashion.
\end{remark}

\subsection{Randomness-oblivious predicates (ROP)}\label{sec:rop_def}

\newcommand{\aware}{{\mathrm{aware}}}
\newcommand{\oblivious}{{\mathrm{obliv}}}

For our application of rectangular PCPs (\cref{sec:pcps_to_rigid}), we would like the decision circuit to depend only on the shared part of the randomness. However, we do not know how to obtain such a PCP (that is also \emph{smooth} and \emph{short}), so we allow the decision circuit to take a limited number of \emph{parity checks} of the entire randomness. Like the decision circuit, the choice of parity checks depends only on the shared part of the randomness.

\begin{definition}[efficient PCP verifiers with $\tau$-ROP]

For $\tau \in [0,1)$, a $(r,q,m,d,t)$-restricted verifier $V$ is said to have the \defemph{$\tau$-randomness-oblivious predicates ($\tau$-ROP)} if the following holds. 

The random coin tosses $R \in \{0,1\}^r$ can be partitioned into two parts 
\[ R = (R_\oblivious, R_\aware) \in \{0,1\}^{(1-\tau)r}\times \{0,1\}^{\tau r},\] 
such that the verifier $V$ on input $x$ of length $n$ and random coins $R$ runs in time $t(n)$, and
\begin{enumerate}
        \item Based only on $R_\aware$:
        \begin{enumerate}
            \item Constructs a \defemph{(decision) predicate} $D \leftarrow V(x; R_\aware)$ of size at most $d(n)$.
            \item Constructs a sequence of \defemph{randomness parity checks} $(C_{1}, \dots, C_{p}) \gets V(x; R_\aware)$, each of which is an affine function on $\bool^{(1-\tau)r} \to \bool$.\footnote{These are affine functions of the oblivious part only, but they encompass parities on \emph{all} of the randomness by including the parity of the aware part in the constant term.}
        \end{enumerate}
        \item Based on all of the randomness $R = \pars{R_\oblivious, R_\aware}$, produces a sequence of $q$ \defemph{proof locations} $I=(i^{(1)},\ldots,i^{(q)})$, where each $i^{(k)} \in [m(n)]$.
    \end{enumerate}
Think of $V$ as representing a probabilistic oracle machine that
queries proof oracle $\pi$ and gets answer symbols $\pi\restrict{I}$,
computes parity checks $P \coloneqq \pars{C_{1}\pars{R_\oblivious}, \dots, C_{p}\pars{R_\oblivious}}$
and accepts iff $D(\pi\restrict{I},P)=1$.
\item
We write $(I,P,D)\getsr V(x)$
to denote the queries, predicate, and parities generated by $V$ on input $x$.
To explicitly mention the random coins $R$, we write $(I,P,D) \gets V(x;R)$. 
\item
We call $p$ the \defemph{parity-check complexity} of $V$.
\end{definition}

We view ROP as a secondary property to RNL and rectangularity, and for simplicity we sometimes omit it from informal discussions (e.g., the title \cref{sec:pcps_to_rigid}). Indeed, in \cref{sec:rop} we show a simple way of adding adding ROP to any PCP while essentially increasing only its decision complexity.

\subsection{A description of a rectangular verifier with ROP}
All new PCP notions that are key to our work deal with a modified view of the run of a PCP verifier based on a partitioning of its randomness. Thus, let us take a moment to provide a streamlined description of a rectangular PCP verifier that has ROP. We hope this description helps the reader picture the new properties of our main PCP verifier, which we eventually construct in \cref{thm:combinedPCP}, and use in our construction of rigid matrices (\cref{sec:pcps_to_rigid}). Specifically, we wish to clarify the dependence of the queries, the decision predicate and the parity checks on the different parts of the randomness.

\begin{remark}[Rectangular verifier with ROP]
    Let $\tau \in \pars{0,1}$, and let $V$ be a $\tau$-rectangular $\pars{r,q,p,\ell^2,d,t}$-verifier with $\tau$-ROP. Assume further that the \defemph{shared} and \defemph{aware} parts of the randomness of $V$ are the same,\footnote{This will be the case in the rest of this work.} such that its randomness $R$ is partitioned as follows:
    \begin{align*}
        R_\oblivious &= \pars{R_\row, R_\column} \\ 
        R_\aware &= R_\shared \\
        R &= \pars{R_\row, R_\column, R_\shared} = \pars{R_\oblivious, R_\aware}.
    \end{align*}
    Since the \defemph{shared} and \defemph{aware} parts of the randomness are the same, we will refer only to the \defemph{shared} part of the randomness.

    The run of $V$ given input $x$ and proof oracle $\pi$ can be described as follows:
    \begin{enumerate}
        \item Sample \defemph{shared} randomness $R_\shared \in \bool^{\tau \cdot r}$. Based on it,
        \begin{enumerate}
            \item Construct a decision predicate $D \coloneqq D\pars{x; R_\shared}$ of size $d$.
            \item Construct randomness parity checks $\pars{C_1,\dots, C_p} \coloneqq \pars{C_1\pars{x; R_\shared}, \dots, C_p\pars{x; R_\shared}}$.
        \end{enumerate}
        \item Sample \defemph{row} randomness $R_\row \in \bool^{\pars{1-\tau}r/2}$. Construct proof row locations
        \begin{equation*}
            i_\row^{\pars{1}} \coloneqq i_\row^{\pars{1}}\pars{x; R_\row, R_\shared}, \dots,  i_\row^{\pars{q}} \coloneqq i_\row^{\pars{q}}\pars{x; R_\row, R_\shared}.
        \end{equation*}

        \item Sample \defemph{column} randomness $R_\column \in \bool^{\pars{1-\tau}r/2}$. Construct proof column locations
        \begin{equation*}
            i_\column^{\pars{1}} \coloneqq i_\column^{\pars{1}}\pars{x; R_\column, R_\shared}, \dots,  i_\column^{\pars{q}} \coloneqq i_\column^{\pars{q}}\pars{x; R_\column, R_\shared}.
        \end{equation*}
        
        \item Compute randomness parity checks $P\coloneqq \pars{C_1\pars{R_\row, R_\column}, \dots, C_p\pars{R_\row, R_\column} }$. Query the proof oracle to obtain $\pi \restrict I \coloneqq \pars{\pi_{i^{(1)}_\row, i^{(1)}_\column}, \dots, \pi_{i^{(q)}_\row, i^{(q)}_\column}}$.
        
        \item Output the result of the computation $D\pars{\pi\restrict{I}, P}$.
        
    \end{enumerate}
    
\end{remark}


\section{From rectangular PCPs to rigid matrices}\label{sec:pcps_to_rigid}

\newcommand{\rank}{\rho}

Alman and Chen \cite{AlmanC2019} show how to construct rigid matrices using efficient, short and smooth PCPs. In this section, we show how efficient, short, smooth and \emph{rectangular} PCPs can be used to obtain a simpler and stronger construction of rigid matrices.

\begin{lemma}\label{claim:pcps_to_rigid}
    Let $\tau\in (0,1)$ and  $\rank:\N\to \N$. 
Let $L \in \NTIME\pars{2^n} \setminus \NTIME\pars{O(2^n/n)}$ be a unary language. Suppose $L$ has a PCP with soundness error $s$ and a $\pars{r,q,\ell^2,d,t}$-restricted verifier $V$ over alphabet $\{0,1\}$ with $\ell\pars{n} = \poly(2^n)$ strictly monotone increasing. Assume further that $V$ is  $\tau$-rectangular and has $\tau$-ROP with parity-check complexity $p$, and that the \emph{shared} and the \emph{aware} parts of the randomness are the same. Lastly, assume that the following inequalities hold:
    \begin{enumerate}
        \item \label{another_assumption_on_parameters} $\frac{1 + \tau}{2} \cdot r + \log (t+\rank) \leq n - \log n$.
        \item \label{assumption_on_parameters} $q + p + r  - \Omega\pars{\frac{\pars{1-\tau}r }{ \log ((q + p) \rank)}} \leq n - \log n$.
        \end{enumerate}
    Then, there is an $\FNP$-machine such that, for infinitely many $N\in \N$, on input $1^N$, outputs an $N\times N$ matrix with $\F_2$ entries which is $\pars{\frac{1-s}{q} \cdot N^2, \rank(\ell^{-1}(N))}$-rigid.
\end{lemma}

To prove \cref{claim:pcps_to_rigid}, we make use of the following fast algorithm that counts the number of $1$'s in a low rank matrix (given its low rank decomposition).
\begin{theorem}[\cite{ChanW2016,AlmanC2019}]
\label{thm:matrix_mul}
Given two matrices $A\in \F_2^{n\times \rho}$ and $B \in \F_2^{\rho \times n}$ where $\rho = n^{o(1)}$, there is a (deterministic) algorithm that computes the number of $1$'s in the matrix $A\cdot B$ in time $T(n,\rho) \coloneqq n^{2 - \Omega\left(\frac{1}{\log \rho}\right)}$.
\end{theorem}

We now prove \cref{claim:pcps_to_rigid}.
\begin{proof}[{Proof of \cref{claim:pcps_to_rigid}}]
Let $V$ be the postulated PCP verifier for $L$.
The $\FNP$-machine computing rigid matrices (infinitely often) runs as follows. On input $1^N$,
\begin{enumerate}
    \item If $N = \ell\pars{n}$ for some $n \in \N$:
    \footnote{ 
    This can be done in time $O(n) = O(\log N)$ by finding the first integer $n$ such that $\ell\pars{n}\ge N$, and verifying that $\ell\pars{n}=N$.}
    \begin{enumerate}
        \item Guess an $N \times N$ matrix denoted by $\pi$.
        \item Emulate $V^\pi \pars{1^n}$ on all possible $2^r$ random coins. If $V$ accepted on all randomness, \defemph{accept} and output $\pi$; else, \defemph{reject}.
    \end{enumerate}
    \item Else ($N \neq \ell\pars{n}$ for any $n$), \defemph{reject}.
\end{enumerate}
This machine runs in time $O\pars{2^r \cdot t} = \poly\pars{2^n} = \poly(\ell(n)) = \poly(N)$, and whenever $N = \ell\pars{n}$ for some $n$ such that $1^n \in L$, one of its non-deterministic guesses lead to acceptance by completeness of the PCP. Note that there could be multiple non-deterministic guesses that lead to acceptance. We show that for infinitely many  $N = \ell\pars{n}$ such that $1^n \in L$, any guessed $\pi$ that leads to acceptance is $\pars{\frac{1-s}{q} \cdot N^2, \rank}$-rigid for $\rank \coloneqq \rank\pars{n} = \rank\pars{\ell^{-1}\pars{N}}$.

Assume towards contradiction that this is not the case.
Then, there exists an $n_0$ such that for any $n\ge n_0$, $1^n \in L$ if and only if there exists a proof $\pi$ (for the verifier $V$) which is $\frac{1-s}{q}$-close to a rank $\rank$ matrix. 
We describe a non-deterministic algorithm that decides $L$ in time $O(2^n / n)$ -- a contradiction. Given input $1^n$,

\begin{enumerate}
	\item Guess matrices $A$ and $B$ of dimensions $\ell \times \rank$ and $\rank \times \ell$ respectively.
	(The right guess is when $A\cdot B$ is $\delta$-close to $\pi$; by smoothness, the acceptance probability of $A\cdot B$ will then be close to that of $\pi$, and the task is reduced to estimating the acceptance probability of $A\cdot B$.)
	\item Compute the acceptance probability of $A\cdot B$ by $V$ as follows. For each sequence of coins in the shared part of the randomness $R_\shared \in \bool^{\tau \cdot r}$:
	
	\begin{enumerate}
		\item \label{compute_decision_circuit} Compute the predicate $D\coloneqq D(R_\shared)$ and randomness parity checks $C_j \coloneqq C_j(R_\shared)$, $j\in [p]$.
		\item \label{compute_mapping_row_and_column} \textbf{Prepare queries into proof:} For each $k\in [q]$,
		\begin{enumerate}
			\item {\bf Prepare left matrices:} Compute the $2^{(1-\tau)r/2} \times \rank$ matrix $A^{(k)}$ whose $R_\row$-th row is just the row indexed by $i_{\row}^{(k)}(R_{\row}, R_{\shared})$ in $A$, for any $R_{\row}\in \{0,1\}^{(1-\tau)r/2}$.
			\item {\bf Prepare right matrices:}  Compute the $\rank \times 2^{(1-\tau)r/2}$ matrix $B^{(k)}$ whose $R_\column$-th column is just the column indexed by $i_{\column}^{(k)}(R_{\column}, R_{\shared})$ in $B$, for any $R_{\column}\in \{0,1\}^{(1-\tau)r/2}$.\\
		\end{enumerate}
		\vspace{-0.5cm}
		Observe that $\pars{A^{(k)} \cdot B^{(k)}}_{R_{\row}, R_{\column}}$ is exactly the $k$-th bit read by the verifier on randomness $(R_{\row},R_{\column}, R_{\shared})$.
		\item \label{compute_parity_checks}\textbf{Prepare randomness parity checks:} For each $j\in [p]$,
		\begin{enumerate}
				    \item Compute the $2^{\pars{1-\tau}r/2} \times 3$ matrix $A^{(q+j)}$ and the $3 \times 2^{\pars{1-\tau}r/2}$ matrix $B^{(q+j)}$, for which $\pars{A^{(q+j)} \cdot B^{(q+j)}}_{R_\row, R_\column} = C_{j}(R_\row, R_\column)$. Such matrices exist and can be computed in time $O(r\cdot 2^{\pars{1-\tau}r/2})$, as described in \Cref{producing_randomness_matrices}.
		\end{enumerate}
		\item \label{fast_counting} {\bf Fast counting:}
		Fourier analysis tells us that there are (unique) coefficients $\braces{\widehat{D}\pars{K}}_{K \subseteq \brackets{q+p}}$ such that for any $y \in \bool^{q+p}$,
		\begin{equation*}
		     D(y_1, \dots, y_{q+p}) = \sum_{K \subseteq \brackets{q+p}} \widehat{D}\pars{K} (-1)^{\bigoplus_{i \in K} y_i}.
		 \end{equation*}
		 In particular, by linearity of expectation, to compute the expected value (which is also the acceptance probability) of $D(y)$ over a random $y$ sampled from some distribution, it suffices to compute the expected value of all parity predicates on over $y$, namely $\Exp_y\brackets{\bigoplus_{i \in K}y_i}$ for all $K \subseteq \brackets{q+p}$.
		 
		This observation is useful because the final task is to compute the acceptance probability of the predicate $D$ on inputs $\pars{A^{(1)} \cdot  B^{(1)}}_{R_\row, R_\column}, \ldots, \pars{A^{(q + p)} \cdot B^{(q + p)}}_{R_\row, R_\column}$ for uniformly random $R_\row$ and $R_\column$. Thus, it suffices to compute the acceptance probability of all parity predicates, i.e., the number of $1$'s in $\bigoplus_{k \in K} A^{(k)} \cdot B^{(k)}$ for each $K \subseteq \brackets{q + p}$.

		 For each $K$, note that $\bigoplus_{k \in K} A^{(k)} \cdot  B^{(k)}$ is the product of a $2^{(1-\tau)r/2}\times (|K|\rank)$ and a $(|K|\rank)\times 2^{(1-\tau)r/2}$ matrix over $\F_2$ (namely, concatenate the rows of the $|K|$ matrices $\{A^{(k)}\}_{k\in K}$ and concatenate the columns of the $|K|$ matrices  $\{B^{(k)}\}_{k\in K}$).
		  Thus, for each $K$, computing the acceptance probability of $\bigoplus_{k \in K} A^{(k)} \cdot B^{(k)}$ can be done with the fast counting algorithm for low-rank matrices of \cref{thm:matrix_mul}. Its runtime is
		  \[
		  T(2^{(1-\tau) r/2}, \pars{q+p} \cdot \rank) = 
			\pars{2^{\pars{1-\tau}r / 2}}^{\pars{2-\Omega \pars{\frac{1}{\log \pars{\pars{q+p} \cdot \rank}}}}} = 2^{\pars{1-\tau} r  - \Omega\pars{\frac{\pars{1-\tau}r} {\log \pars{\pars{q+p} \cdot \rank} }}}
			\]
	\end{enumerate}
	\item We have thus computed the acceptance probability of $A\cdot B$ by the verifier $V$. If this probability is at least the soundness error $s$, decide that the input $1^n$ is in $L$. Otherwise, decide that the input is not in $L$.
\end{enumerate}
We claim that the algorithm correctly decides $L$, for input length $n\ge n_0$. Indeed, if $1^n \in L$ then, when guessing $A$ and $B$ such that $A\cdot B$ is $(\frac{1-s}{q})$-close to the correct proof $\pi$ for $1^n$ (such $A$ and $B$ exist by our assumption towards contradiction), smoothness of the verifier $V$ implies that its acceptance probability is at least $1-q\cdot(\frac{1-s}{q}) = s$. On the other hand, if $1^n \notin L$, then soundness of $V$ implies any guessed $A$ and $B$ leads to rejection with probability strictly less than $s$.
Since $n_0$ is constant we can hard-wire the values of $L$ on $1^{n}$ for $n<n_0$ so that the algorithm correctly decides $L$ on all inputs.

As for its runtime, observe that \cref{compute_decision_circuit} takes time $O(t)$, \cref{compute_mapping_row_and_column} takes time $O( 2^{(1-\tau)r/2}\cdot (t+\rank))$ and \cref{compute_parity_checks} takes time $O(r\cdot 2^{(1-\tau)r/2})$. Since $t\geq r$, these are dominated by $O(2^{(1-\tau)r/2}\cdot (t+\rank))$. Therefore the runtime of the algorithm is
\begin{align*}
	&O\pars{2^{\tau \cdot r} \cdot \left( 
	2^{\frac{(1-\tau)r}{2}}\cdot (t+\rank)
	\;+\; 2^{q+p}\cdot 2^{\pars{1-\tau} r - \Omega\pars{\frac{\pars{1-\tau}r}{\log \pars{\pars{q+p} \cdot \rank}}}}
	\right)}\\
	&=
	O\pars{2^{\frac{(1 + \tau)r}{2}} \cdot  (t+\rank)  \;+\; 2^{q + p   + r  - \Omega\pars{\frac{\pars{1-\tau}r}{\log \pars{\pars{q+p} \cdot \rank}}}}}.
\end{align*}
By the assumption on the parameters of the PCP, this is at most $O(2^n/n)$ -- a contradiction.
\end{proof}

Now that we formalized a connection between rectangular PCPs and rigid matrices, let us introduce the PCP that we construct in the remainder of this work, and show how it implies the rigid matrix construction asserted in \cref{thm:main}.

\begin{theorem}[{\cref{thm:combinedPCP}, instantiated}]
\label{thm:combinedPCP_simplified}
For any $L \in \NTIME\pars{2^n}$, and constants $s\in (0,1/2)$ and $\tau \in (0,1)$, $L$ has a PCP verifier over alphabet $\{0,1\}$ with the following parameters:
\begin{itemize}
    \item Randomness complexity $r\pars{n} = n +  O(\log n)$
    \item Proof length $\ell(n)^2 = 2^n\cdot \poly(n)$.
    \item Soundness error $s$.
    \item Decision,  query and parity-check complexities all $O(1)$.
    \item Verifier runtime $t(n) = 2^{O(\tau n)}$.
    \item The verifier is $\tau$-rectangular and has $\tau$-ROP. Furthermore, the shared and the aware parts of the randomness is the same.
\end{itemize}
\end{theorem}

\cref{thm:combinedPCP_simplified} is obtained by instantiating \cref{thm:combinedPCP} with parameters $T(n) \coloneqq 2^n$ and $m \coloneqq \Omega(1/\tau)$. Next, we restate the rigid matrix construction asserted in \cref{thm:main} and reckon that it is obtained by combining \cref{thm:combinedPCP_simplified} and \cref{claim:pcps_to_rigid}.

\begin{corollary}[{\cref{thm:main}, restated}]\label{cor:main_thm}
    \maintheoremstatement
\end{corollary}

\begin{proof}
From \cite{Zak1983}, there exists a unary language $L \in \NTIME\pars{2^n} \setminus \NTIME\pars{O(2^n/n)}$.  Let $L$ be any such language.
Fix $\rank(n) \coloneqq 2^{n/(K\log n)}$ for a large enough constant $K$ to be determined later. We verify that the parameters of the PCP of \cref{thm:combinedPCP_simplified}, for a sufficiently small $\tau\in (0,1)$, satisfy all the conditions from \cref{claim:pcps_to_rigid} for this $\rank$. Let $q$ and $p$ denote the query complexity and parity-check complexity of the PCP respectively. Set $\delta = (1-s)/q$. Now, for small enough $\tau$,
$$\pars{\tfrac{1 + \tau}{2}} \cdot r + \log (\rank + t) \leq (1/2 + O(\tau))n + O(n/\log n) < n - \log n,$$
as required in \cref{another_assumption_on_parameters}. Also, the parameters satisfy \cref{assumption_on_parameters}, because
$$q + p + r  - \Omega\pars{\frac{\pars{1-\tau}r }{ \log ((q + p) \rank)}}  \le n  + O(\log n) - \Omega(K \log n) < n - \log n,$$ where the last inequality holds for a suitable choice of the constant $K$.
As the proof length is $\ell(n)^2 = 2^n\cdot \poly(\log n)$, we have $\ell^{-1}(N) = \Theta(\log N)$. Therefore, $\rho(\ell^{-1}(N)) = 2^{\log N/\Omega(\log \log N)}$ and the corollary follows from \cref{claim:pcps_to_rigid}.
\end{proof}

\begin{remark}
	The only bottleneck preventing \cref{claim:pcps_to_rigid} from giving rigid matrices for polynomial ranks, $\rank=N^{\Omega(1)}$, via \cref{thm:combinedPCP_simplified} is the runtime of the counting algorithm used in \Cref{fast_counting}. That is, such results would be obtained if there was an algorithm for counting the number of nonzero entries in a rank $N^{\Omega(1)}$ matrix (of dimensions $N \times N$) that ran in time slightly better than $O\pars{N^2}$. Our reduction would go through even if the algorithm's answer was only \emph{approximately} correct (while losing the respective approximation factor in the distance of the resulting matrices from low rank ones). In fact, this seems to be the only bottleneck even up to rank $\rank = N^{1-O(\tau)}$. 
\end{remark}


\section{From RNL to smooth and rectangular PCPs}\label{sec:smoothification}
\newcommand{\idx}{i}

\newcommand{\new}{\mathrm{new}}
\newcommand{\old}{\mathrm{old}}
\newcommand{\oldpi}{{\pi_\old}}
\newcommand{\oldV}{{V_\old}}
\newcommand{\newpi}{{\pi_\new}}
\newcommand{\newV}{{V_\new}}

In this section, we show how any PCP verifier with RNL can be made into a smooth and rectangular PCP. This conversion preserves the ROP.

\begin{theorem}\label{thm:smoothification}

Suppose $L$ has a PCP with verifier $\oldV$ as described in \cref{fig:smoothification}, and $\tau$-RNL and $\tau$-ROP such that the \defemph{shared} and  \defemph{aware} parts of the randomness are the same.
Then for any $\mu \in \pars{0,1}$,
$L$ has a PCP with verifier $\newV$ as described in \cref{fig:smoothification} which is \emph{smooth}, $\tau$-rectangular and $\tau$-ROP such that the \defemph{shared} and  \defemph{aware} parts of the randomness are the same.

\begin{table}[ht]
    \centering
    \begin{tabular}{|c|c|c|}
        \hline 
        Complexity & $\oldV$ & $\newV$\tdelim
        \hline 
        Alphabet size &
        $2$ &
        $2$ \tdelim
        Soundness error & $s$ & $s + \mu$ \tdelim
        Randomness & $r$ & $r$ \tdelim
        Query & $q$ & $\poly\pars{q / \mu}$ \tdelim
        Parity-check & $p$ & $p$ \tdelim
        Proof length & $m$ & $2^r \cdot q$ \tdelim
        Decision & $d$ & $d + \poly\pars{q / \mu}$ \tdelim
        Runtime & $t$ & $t + q \cdot \poly\pars{t_\RNL}$ \tdelim
        \end{tabular}
    \caption{The complexities of the original $\oldV$ and the smooth verifier $\newV$. $t_\RNL$ is the running time of the row and column neighbor-listing agents of $\oldV$.}
    \label{fig:smoothification}
\end{table}

\end{theorem}

The smooth proof system of \cref{thm:smoothification} utilizes an explicit construction of \emph{sampler graphs}, defined and stated next.
\begin{definition}[Sampler graph]
    Fix $\alpha \in \brackets{0,1}$. Graph $G = \pars{V,E}$ is an $\alpha$-sampler if for every $S \subseteq V$,
    \begin{equation*}
        \Pr_{v \in V}\brackets{ \abs{ \frac{|S|}{|V|} - \frac{|\Gamma\pars{v} \cap S|}{|\Gamma\pars{v}|}} > \alpha} < \alpha.
    \end{equation*}
\end{definition}

\begin{fact}[{\cite[Section~5.1]{Goldreich2011-samp}}]\label{fact:samplers}
    There exists an algorithm that given an integer $n$ and $\alpha \in \pars{0,1}$, constructs a $(4/\alpha^4)$-regular graph on $n$ vertices which is an $\alpha$-sampler in time $\poly\pars{n}$.
\end{fact}

\subsection{The smooth and rectangular PCP}
The smooth and rectangular PCP verifier is obtained by applying the degree reduction transformation of \cite[Theorem 5.1]{DinurH2013}. We restate this transformation with syntactic changes that will be helpful for showing rectangularity.\footnote{Our presentation is from the ``proof systems'' perspective of PCP verifiers, rather than the ``label cover'' perspective given in \cite[Theorem 5.1]{DinurH2013}}

Let $\oldV$ be the verifier postulated in \Cref{thm:smoothification} and denote by $\newV$ the new, smooth and rectangular, verifier. We start by describing the proofs expected by $\newV$.

\paragraph{Proofs in the new PCP system.}
New proofs are of length $2^r \cdot q$, which we think of as indexed over $\bool^r \times \brackets{q}$. Each location in the new proof corresponds to a full configuration of the original verifier. A correct proof in the original proof system is transformed to a correct proof in the new one by writing in the $\pars{R,k}$-th location the answer of the original correct proof to the $k$-th query when randomness $R$ is sampled. That is, for an input $x \in L$ and a correct proof $\pi$ for $\oldV$ (i.e., one that is accepted w.p. 1), the $\pars{R,k}$-th location of the correct proof for $\newV$ will have the answer of $\pi$ to the $k$-th query issued by $\oldV$ upon sampling random coin sequence $R$. Notice that in a correct proof for the new verifier, any two locations in the new proof corresponding to neighboring configurations (see \Cref{def:neighbor}) should take the same value.

\paragraph{The new verifier.}
The basic idea is for the new verifier to emulate the original one: when the original samples coin sequence $R$, the new one queries locations $\pars{R,1},\dots,\pars{R,q}$ in the new proof. However, if the original verifier queried the same location $i \in \brackets{\ell}$ for two different (neighboring) configurations $\pars{R,k}$ and $\pars{R^\prime, k^\prime}$, a new (``cheating'') proof could be \emph{inconsistent} in its answers, using this inconsistency to cause the new verifier to accept when the original would not.

Thus, consistency between neighboring configurations must be checked. To guarantee smoothness and preserve the randomness of the new verifier, consistency is checked only between certain neighboring configurations, and not all. Namely, neighboring configurations are connected by a $O(\mu/q)$-sampler graph of constant degree, with edges on the sampler corresponding to consistency tests. Two configurations that are adjacent on the sampler are said to be \defemph{sampler-neighbors}, which is a stronger condition than being neighbors as per \Cref{def:neighbor}.

In fact, both consistency and the original PCP verification are done in one fell swoop: when emulating the original verifier, the new verifier replaces a query to $\pars{R,k}$ with queries to its sampler-neighborhood, and checks its consistency. A sampler graph guarantees that inconsistency between neighboring configurations is reflected by this test w.h.p., so severely inconsistent proofs are rejected by the new verifier. On the other hand, regularity of the sampler implies smoothness, and its constant degree incurs only a constant blowup to the number of queries. 

The point of this theorem is in showing rectangularity of the new verifier. Specifically, we ought to show how construction of the sampler and sampler-neighborhoods can be done \emph{rectangularly}. That is, for any location $\pars{R_\row,R_\column,R_\shared, k}$, it is not enough to find all other sampler-neighboring $\pars{R^\prime_\row,R^\prime_\column, R^\prime_\shared, k^\prime}$; it should be the case that the 
\emph{row-part} of the sampler-neighbors can be found based on $\pars{R_\row, R_\shared, k}$. Similarly, the 
\emph{column-part} of the sampler-neighbors should be found only from $\pars{R_\column, R_\shared, k}$. 

We clarify what we mean by ``row-part'' and ``column-part'' of a query. The new proof (of length $2^r \cdot q$) can be thought of as a square matrix as follows: Fix a location $\pars{R_\row, R_\column, R_\sharedrow, R_\sharedcolumn, k}$ in the new proof. Split $k$ into $k_\row$ and $k_\column$. The rows of the matrix are indexed by $\pars{R_\row, R_\sharedrow, k_\row}$, and the columns are indexed by $\pars{R_\column, R_\sharedcolumn, k_\column}$. Indeed, with this definitions, it is possible to find the row-parts (resp., column-parts) of the sampler-neighbors based on $\pars{R_\row, R_\shared, k}$ (resp., $\pars{R_\row, R_\shared, k}$) thanks to RNL.

Following is a detailed description of this construction.

\newcommand{\answer}{\mathtt{answer}}

\begin{algorithm}\label{alg:new_verifier}
Fix the original verifier $\oldV$.

\begin{enumerate}
	\item Sample a coin sequence $R$.
	\item \label{find_neighbors} For each $k \in \brackets{q}$, construct the sampler of the neighborhood of $\pars{R,k}$ and check consistency of its sampler-neighborhood \emph{in a rectangular way} as follows: Denote the randomness partition by of $R$ by $\pars{R_\row,R_\column,R_\shared}$.
	\begin{enumerate}
		\item \label{find_row_part} Find the ``row parts'': Compute $L_\row \coloneqq A_\row\pars{R_\row, R_\shared, k}$ where $A_\row$ is the neighbor listing agent. Construct a canonical $(\mu/2q)$-sampler on the set of $|L_{\row}|$ vertices, one corresponding to every entry in the list $L_{\row}$. From the index of $(R,k)$ in the list $L_{\row}$, find the indices of the sampler-neighbors of $\pars{R, k}$, and output their ``row-part'' $\pars{R^\prime_\row,R^\prime_\sharedrow, k^{\prime}_{\row}}$.
		In addition, output the row-part of $\pars{R,k}$.
		\item \label{find_column_part} Find the ``column part'': Similarly, compute $L_\column \coloneqq A_\column\pars{R_\column, R_\shared, k}$, construct a canonical $(\mu/2q)$-sampler on the set of $|L_{\column}|$ vertices, find the indices of the sampler-neighbors of $\pars{R, k}$, and output their ``column-part''  $\pars{R^{\prime}_\column,R^{\prime}_\sharedcolumn, k^{\prime}_{\column}}$.
		In addition, output the column-part of $\pars{R,k}$.

	\end{enumerate}
	\item 
	Let $\agla-1$ denote the degree of the sampler above. Note that we are making $\agla\cdot q$ many queries.
	Feed the $\agla\cdot q$ bits queried from the proof $\pi_{\new}$ to a circuit that first checks consistency between every sampler-neighborhood. That is, it check that in each of the $q$ blocks of $\agla$ bits, all the $\agla$ bits are equal. If an inconsistency is spotted, the circuit immediately rejects. Otherwise, feed the first bit in every block to the decision circuit of the original verifier $\oldV$ (along with the $p$ parity-checks on the randomness) and output its answer.
\end{enumerate}

\end{algorithm}

\subsection{Proof of \Cref{thm:smoothification}}

\paragraph{Rectangularity.}

The randomness of the new verifier is split exactly the same as the original verifier into $R_\row$, $R_\column$ and $R_\shared = (R_{\sharedrow}, R_{\sharedcolumn})$. 
It follows from the description of the algorithm that $R_\row$ and $R_\shared$ determine the row-part of each of the $q\cdot \agla$ queries. Similarly, $R_\column$ and $R_\shared$ determine the column-part of each of the $q\cdot \agla$ queries.

\paragraph{ROP.}
The new decision predicate can be implemented by taking a circuit that checks equality on each of the $q$ sampler-neighborhoods constructed in \cref{find_neighbors}, and ANDing its answer with the output of original decision circuit (fed an arbitrary representative of each sampler-neighborhood, as well as the randomness parity checks). Therefore, the $\tau$-ROP is preserved.

\paragraph{Query and decision complexities.}
By \cref{fact:samplers}, the size of each sampler-neighborhood (in a $\pars{\mu / 2q}$-sampler) is $\poly\pars{q / \mu}$. A sampler-neighborhood is queried for each of the $q$ original, so the query complexity $\pars{q / \mu}$. As described in the ROP analysis, the new decision circuit can be obtained by ANDing the original decision circuit (of size $d$) to $\poly\pars{q / \mu}$ equality checks. Thus, the size of the new circuit is $d + \poly\pars{q / \mu}$.

\paragraph{Smoothness.}
The fact that the samplers we construct are $\agla$-regular graphs implies that every location in the new proof is read with exactly the same probability (see \cref{sec:smooth_smoothness} for more details).

\paragraph{Soundness.}
One way to see soundness (as well as smoothness) would be to observe \Cref{alg:new_verifier} is the same as the verifier of \cite[Theorem 5.1]{DinurH2013}, and is therefore sound (and smooth). Since the latter theorem and its proof are described in the ``label cover'' view of PCPs whereas our work takes the ``proof systems'' view, we present an alternative proof in the latter view in \cref{sec:smooth_soundness}.

\paragraph{Runtime complexity.}
The new verifier emulates the original one. In addition, for each of the $q$ queries it invokes RNL agents and finds a neighborhood in the sampler. Invoking RNL agents takes $t_\RNL$ time. Constructing the explicit sampler on the configuration's neighborhood and finding its sampler-neighborhood takes time at most $\poly\pars{t_\RNL}$ time (we upper bound the size of each list with the runtime of each agent).

\section{A many-query robust PCP with RNL} \label{section:outer}

In this section, we prove that the Reed--Muller-based PCP of Ben-Sasson \etal~\cite{BenSassonGHSV2006,BenSassonGHSV2005} has RNL. This PCP issues many queries, but is {\em robust} -- which will come useful later in the composition stage (\cref{section:composition}) to reduce its query complexity. In particular, we modify the many-query robust PCP of Ben-Sasson \etal~\cite{BenSassonGHSV2006,BenSassonGHSV2005} to obtain the following PCP with RNL.\footnote{For the sake of simplicity, instead of working with verifier specifications (as in we work with verifiers instead of verifier specifications (\cite{BenSassonGHSV2005} use the latter). This costs us an extra $q(n)$ factor in the running time. However this loss of $q(n)$ has no effect on our application to rigid matrices.}

\newcommand{\AO}{\tilde{A}}

\begin{theorem}[{Strengthening of \cite[Theorem 3.1]{BenSassonGHSV2006} and \cite{BenSassonGHSV2005}, simplified}]\label{thm:outerBGHSV}

Suppose that $L$ is a language in
$\NTIME(T(n))$ for some non-decreasing function $T \colon \N \rightarrow \N$. There exists a universal constant $c$ such for all odd integers $m \in \N$ and $s \in \pars{0,1/2}$ satisfying $T(n)^{1/m} \geq m^{cm}/s^6$, $L$ has a robust PCP with the following parameters:
\begin{enumerate}
\item \label{bghsv_item_1} Alphabet $\{0,1\}$.
\item \label{bghsv_item_2} Randomness complexity $r\pars{n} = (1-\frac{1}{m})\log T(n) + O(m \log \log T(n))
+ O(\log(1/s))$.
\item \label{bghsv_item_3} Decision and  Query complexity $d(n) = q\pars{n} = T(n)^{1/m} \cdot \poly(\log T(n), 1/s)$.
\item \label{bghsv_item_4} Robust soundness error $s$ with robustness parameter $\Theta(s)$.
\item \label{bghsv_item_5} Runtime complexity $t(n) = q(n) \cdot \poly(n, \log  T(n))$.
\item \label{bghsv_item_6} The PCP verifier has $\tau$-RNL with running time $t_\RNL(n)$ where
\begin{align*}
\tau \cdot r(n) &=r_\shared = \frac{4}{m}\log T(n) +  O(m\log \log T(n)) 
+ O(\log(1/s)),\\
t_\RNL(n) &= \poly(\log T(n)).
\end{align*}%
\end{enumerate}
\end{theorem}
\begin{remark}
	\cref{bghsv_item_1,bghsv_item_2,bghsv_item_3,bghsv_item_4} are exactly as in the statement of \cite[Theorem 3.1]{BenSassonGHSV2006}, the outer robust PCP construction of Ben-Sasson~\etal while Item 5 (the verifier running time) is obtained by the efficient PCP verifiers of Ben-Sasson~\etal~\cite{BenSassonGHSV2005}. The main difference between the two works of Ben-Sasson~\etal~\cite{BenSassonGHSV2006,BenSassonGHSV2005} is that the latter uses a reduction from Succinct-SAT to Succinct-Multivariate-Algebraic-CSP. To show that these PCPs have RNL, we need to analyze the query and predicate of the corresponding PCP verifiers. Since these are almost identical in both the constructions (i.e., the original robust PCP construction and the subsequent efficient version of it), we work with the robust PCP of \cite{BenSassonGHSV2006} and just observe that these modifications can be carried out efficiently as in \cite{BenSassonGHSV2005}.

\end{remark}


\begin{remark}\label{rem:CLW}
    We remark that the runtime complexity of the \cite{BenSassonGHSV2006,BenSassonGHSV2005} verifier is in fact $q(n)\cdot \polylog T(n) + O(n)$ (as observed by Chen, Lyu and Williams ~\cite{ChenLW2020}). This improvement is obtained by constructing a robust PCP of proximity variant of \cref{thm:outerBGHSV} (which constructs only a robust PCP), and then converting it to to a robust PCP using a standard transformation based on linear time-encodable error-correct codes. However, this improvement is not needed for our main result.
\end{remark}

The robust PCP verifiers of Ben-Sasson~\etal~\cite{BenSassonGHSV2006,BenSassonGHSV2005} already have the properties listed in \cref{bghsv_item_1,bghsv_item_2,bghsv_item_3,bghsv_item_4,bghsv_item_5}. In this section we show that it also has RNL as stated in \cref{bghsv_item_6} above. This is done in two steps: First (\Cref{sec:rPCP-large}), we show that the robust PCP from \cite[Section 8.2.1]{BenSassonGHSV2006} has RNL, albeit over a large alphabet. Second (\Cref{sec:alp_red}), we reduce the alphabet size to binary while preserving RNL.

\subsection{The robust RNL PCP verifier over a large alphabet}\label{sec:rPCP-large}

\begin{lemma}
\label{lemma:BGHSV_prop_large_alphabet}
	There exists a robust PCP verifier for a language in $\NTIME(T(n))$ with the properties mentioned in \Cref{thm:outerBGHSV}, but over an alphabet of size  $2^{\polylog\pars{T(n)}}$ (instead of the Boolean alphabet).
\end{lemma}

For convenience, we remind the reader of the PCP verifier from \cite{BenSassonGHSV2005}, following its presentation in \cite[Section 8.2.1]{BenSassonGHSV2006} (see \cref{remark:bghsv_vs_bghsv_field_size} for a minor difference). That same work shows completeness and robust soundness of this verifier. Thus, we need only to show that this PCP has RNL, so it suffices for us to recall its \emph{query patterns} without detailing the way its \emph{decision (predicate)} is made based on these queries.

\begin{algorithm}[Verifier query pattern \protect{\cite[Section 8.2.1]{BenSassonGHSV2006}}]\label{alg:bghsv_query_pattern}
	Let $\F$ be a field of size $|\F| = T(n)^{1/m}\cdot \poly\log T(n)$ where $m$ is an odd integer.
	The proof oracle is a map $\Pi \colon \F^m \rightarrow \F^d$ where $d = m\cdot \poly(\log T(n))$. Let $\shift \colon \F^m \rightarrow \F^m$ denote the linear transformation that cyclically shifts each coordinate to the left; i.e., $\shift(x_1,\ldots,x_{m}) \coloneqq  (x_2, x_3, \ldots, x_{m}, x_1)$. 
	
	The queries will be based on \emph{lines} through $\F^m$, and their shifts. Recall that the line $\calL$ with intercept $x \in \F^m$ and direction $y\in \F^m$ is the set $\calL \coloneqq \braces{ x + ty \colon t \in \F }$. The verifier samples lines from two distributions:
	
	\begin{itemize}
	    \item A \defemph{first-axis parallel line}, in which $y = \pars{1,0,\dots,0}$, and $x$ is sampled uniformly at random from $\{0\}\times \F^{m-1}$.
	    \item A \defemph{canonical pseudorandom line} is chosen by first uniformly sampling a direction $y$ from a $\lambda$-biased set $S_\lambda \subseteq \F^m$. The direction $y$ partitions the space $\F^m$ into $|\F|^{m-1}$ parallel lines, and one of these lines is chosen uniformly at random. Note that this requires $\log(|\bset_\lambda|) + (m-1)\log(|\F|)$ random bits.
	    Without loss of generality, we assume that all $y\in \bset_{\lambda}$ have $y_1\neq 0$.%
	    \footnote{To obtain such a $\lambda$-biased set, construct a $\lambda/2$-biased set and remove all $y$'s with $y_1=0$ from it.} Thus, similarly to first-axis parallel lines, $x$ can be sampled uniformly at random from $\{0\}\times \F^{m-1}$.
	\end{itemize}

\end{algorithm}

\begin{remark}[Differences in the query pattern of \cite{BenSassonGHSV2006} and \cite{BenSassonGHSV2005}] \label{remark:bghsv_vs_bghsv_field_size}
    The only difference in the robust PCP construction of \cite{BenSassonGHSV2006} and its efficient counterpart in \cite{BenSassonGHSV2005} is the reduction from $\NTIME(T(N))$ to (succinct) Multivariate-Algebraic-CSP (which in turn uses the Pippenger--Fisher reduction \cite{PippengerF1979})~\cite[Definition~6.3, Theorem~6.4]{BenSassonGHSV2005}. This causes the field size to increase from $O(m^2\cdot T(n)^{1/m})$ to $T(n)^{1/m}\cdot \poly\log T(n)$.
\end{remark}

\begin{proof}[Proof of \Cref{lemma:BGHSV_prop_large_alphabet}]

We prove that the verifier of \cref{alg:bghsv_query_pattern} has RNL. First, note that the verifier uses a total of $\log(|\bset_\lambda|) + (m-1)\log(|\F|)$ random bits to sample a first-axis parallel line $\calL_0$ and a canonical pseudorandom line $\calL_1$, reusing the random bits between these two lines. It then queries the proof oracle on $\calL_0$, $\shift(\calL_0)$,  $\calL_1$,  and $\shift(\calL_1)$.\footnote{The size of the set $S_\lambda$ and a detailed description of how we derived this query pattern based on the tests of the verifier of \cite[Section 8.2.1]{BenSassonGHSV2006} can be found in \Cref{appendix:bghsv_query_details}.}

The randomness used for sampling $x$ is partitioned into $(m-1)$ parts of equal length, denoted by $(R_2, R_3, \ldots, R_m)$ where $|R_i| = \log(|\F|)$, and $R_i$ determines $x_i$ for each $i \in \braces{2, \dots, m}$ (recall that $x_1 = 0$ always). 
The randomness used to sample a direction from $\bset_\lambda$ is denoted by $R_y$. Recall that there are two-types of lines, canonical and axis-parallel lines. In both cases the direction of the line $y$ is a function of the bits $R_y$ (in the axis-parallel line the direction is just the constant function).

The \defemph{row}, \defemph{column} and \defemph{shared} parts of the randomness are portrayed in \cref{fig:sep_partition}. Formally,
\begin{align*}
	R_\row &\coloneqq \pars{R_3, \ldots, R_{(m-1)/2}}  \\
	R_\column &\coloneqq \pars{R_{(m+5)/2}, \dots, R_{m-1}} \\
	R_\sharedrow &\coloneqq (R_2,R_{(m+1)/2}) \\
	R_\sharedcolumn &\coloneqq (R_{(m+3)/2},R_{m})\\
	R_\shared &\coloneqq \pars{R_\sharedrow, R_\sharedcolumn, R_y},
\end{align*}
where $R_y$ is divided arbitrarily between $R_\sharedrow$ and $R_\sharedcolumn$. Thus, the randomness has parts of length $r_\row \coloneqq \abs{R_\row}$, $r_\column \coloneqq \abs{R_\column}$, $r_\shared \coloneqq \abs{R_\shared}$, with total randomness $r = r_\row + r_\column + r_\shared$ and $\tau \cdot r = r_\shared$.

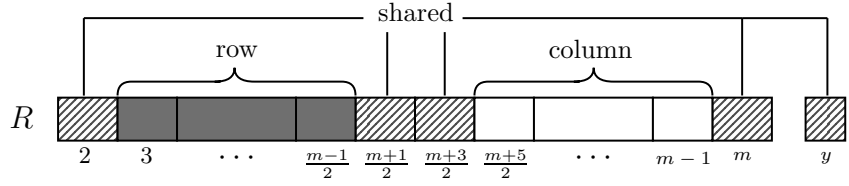
\begin{figure}
    \centering

 
\tikzset{
pattern size/.store in=\mcSize, 
pattern size = 5pt,
pattern thickness/.store in=\mcThickness, 
pattern thickness = 0.3pt,
pattern radius/.store in=\mcRadius, 
pattern radius = 1pt}
\makeatletter
\pgfutil@ifundefined{pgf@pattern@name@_6rkh1tm79}{
\pgfdeclarepatternformonly[\mcThickness,\mcSize]{_6rkh1tm79}
{\pgfqpoint{0pt}{0pt}}
{\pgfpoint{\mcSize+\mcThickness}{\mcSize+\mcThickness}}
{\pgfpoint{\mcSize}{\mcSize}}
{
\pgfsetcolor{\tikz@pattern@color}
\pgfsetlinewidth{\mcThickness}
\pgfpathmoveto{\pgfqpoint{0pt}{0pt}}
\pgfpathlineto{\pgfpoint{\mcSize+\mcThickness}{\mcSize+\mcThickness}}
\pgfusepath{stroke}
}}
\makeatother

 
\tikzset{
pattern size/.store in=\mcSize, 
pattern size = 5pt,
pattern thickness/.store in=\mcThickness, 
pattern thickness = 0.3pt,
pattern radius/.store in=\mcRadius, 
pattern radius = 1pt}
\makeatletter
\pgfutil@ifundefined{pgf@pattern@name@_zdgom59gv}{
\pgfdeclarepatternformonly[\mcThickness,\mcSize]{_zdgom59gv}
{\pgfqpoint{0pt}{0pt}}
{\pgfpoint{\mcSize+\mcThickness}{\mcSize+\mcThickness}}
{\pgfpoint{\mcSize}{\mcSize}}
{
\pgfsetcolor{\tikz@pattern@color}
\pgfsetlinewidth{\mcThickness}
\pgfpathmoveto{\pgfqpoint{0pt}{0pt}}
\pgfpathlineto{\pgfpoint{\mcSize+\mcThickness}{\mcSize+\mcThickness}}
\pgfusepath{stroke}
}}
\makeatother

 
\tikzset{
pattern size/.store in=\mcSize, 
pattern size = 5pt,
pattern thickness/.store in=\mcThickness, 
pattern thickness = 0.3pt,
pattern radius/.store in=\mcRadius, 
pattern radius = 1pt}
\makeatletter
\pgfutil@ifundefined{pgf@pattern@name@_1ccy1g5z9}{
\pgfdeclarepatternformonly[\mcThickness,\mcSize]{_1ccy1g5z9}
{\pgfqpoint{0pt}{0pt}}
{\pgfpoint{\mcSize+\mcThickness}{\mcSize+\mcThickness}}
{\pgfpoint{\mcSize}{\mcSize}}
{
\pgfsetcolor{\tikz@pattern@color}
\pgfsetlinewidth{\mcThickness}
\pgfpathmoveto{\pgfqpoint{0pt}{0pt}}
\pgfpathlineto{\pgfpoint{\mcSize+\mcThickness}{\mcSize+\mcThickness}}
\pgfusepath{stroke}
}}
\makeatother

 
\tikzset{
pattern size/.store in=\mcSize, 
pattern size = 5pt,
pattern thickness/.store in=\mcThickness, 
pattern thickness = 0.3pt,
pattern radius/.store in=\mcRadius, 
pattern radius = 1pt}
\makeatletter
\pgfutil@ifundefined{pgf@pattern@name@_yyrml1dyb}{
\pgfdeclarepatternformonly[\mcThickness,\mcSize]{_yyrml1dyb}
{\pgfqpoint{0pt}{0pt}}
{\pgfpoint{\mcSize+\mcThickness}{\mcSize+\mcThickness}}
{\pgfpoint{\mcSize}{\mcSize}}
{
\pgfsetcolor{\tikz@pattern@color}
\pgfsetlinewidth{\mcThickness}
\pgfpathmoveto{\pgfqpoint{0pt}{0pt}}
\pgfpathlineto{\pgfpoint{\mcSize+\mcThickness}{\mcSize+\mcThickness}}
\pgfusepath{stroke}
}}
\makeatother

 
\tikzset{
pattern size/.store in=\mcSize, 
pattern size = 5pt,
pattern thickness/.store in=\mcThickness, 
pattern thickness = 0.3pt,
pattern radius/.store in=\mcRadius, 
pattern radius = 1pt}
\makeatletter
\pgfutil@ifundefined{pgf@pattern@name@_swlg460vj}{
\pgfdeclarepatternformonly[\mcThickness,\mcSize]{_swlg460vj}
{\pgfqpoint{0pt}{0pt}}
{\pgfpoint{\mcSize+\mcThickness}{\mcSize+\mcThickness}}
{\pgfpoint{\mcSize}{\mcSize}}
{
\pgfsetcolor{\tikz@pattern@color}
\pgfsetlinewidth{\mcThickness}
\pgfpathmoveto{\pgfqpoint{0pt}{0pt}}
\pgfpathlineto{\pgfpoint{\mcSize+\mcThickness}{\mcSize+\mcThickness}}
\pgfusepath{stroke}
}}
\makeatother
\tikzset{every picture/.style={line width=0.75pt}} 

\begin{tikzpicture}[x=0.75pt,y=0.75pt,yscale=-1,xscale=1]

\draw  [pattern=_6rkh1tm79,pattern size=3pt,pattern thickness=0.75pt,pattern radius=0pt, pattern color={rgb, 255:red, 74; green, 74; blue, 74}] (60,102.83) -- (90,102.83) -- (90,124.53) -- (60,124.53) -- cycle ;
\draw  [pattern=_zdgom59gv,pattern size=3pt,pattern thickness=0.75pt,pattern radius=0pt, pattern color={rgb, 255:red, 74; green, 74; blue, 74}] (210,102.83) -- (240,102.83) -- (240,124.53) -- (210,124.53) -- cycle ;
\draw  [pattern=_1ccy1g5z9,pattern size=3pt,pattern thickness=0.75pt,pattern radius=0pt, pattern color={rgb, 255:red, 74; green, 74; blue, 74}] (240,102.83) -- (270,102.83) -- (270,124.53) -- (240,124.53) -- cycle ;
\draw  [pattern=_yyrml1dyb,pattern size=3pt,pattern thickness=0.75pt,pattern radius=0pt, pattern color={rgb, 255:red, 74; green, 74; blue, 74}] (390,102.83) -- (420,102.83) -- (420,124.53) -- (390,124.53) -- cycle ;
\draw  [pattern=_swlg460vj,pattern size=3pt,pattern thickness=0.75pt,pattern radius=0pt, pattern color={rgb, 255:red, 74; green, 74; blue, 74}] (437,102.83) -- (457,102.83) -- (457,124.53) -- (437,124.53) -- cycle ;

\draw  [fill={rgb, 255:red, 111; green, 111; blue, 111 }  ,fill opacity=1 ] (90,102.83) -- (120,102.83) -- (120,124.53) -- (90,124.53) -- cycle ;
\draw  [fill={rgb, 255:red, 111; green, 111; blue, 111 }  ,fill opacity=1 ] (180,102.83) -- (210,102.83) -- (210,124.53) -- (180,124.53) -- cycle ;
\draw  [fill={rgb, 255:red, 111; green, 111; blue, 111 }  ,fill opacity=1 ] (120,102.83) -- (180,102.83) -- (180,124.53) -- (120,124.53) -- cycle ;
\draw  [fill={rgb, 255:red, 255; green, 255; blue, 255 }  ,fill opacity=1 ] (270,102.83) -- (300,102.83) -- (300,124.53) -- (270,124.53) -- cycle ;
\draw  [fill={rgb, 255:red, 255; green, 255; blue, 255 }  ,fill opacity=1 ] (300,102.83) -- (360,102.83) -- (360,124.53) -- (300,124.53) -- cycle ;
\draw  [fill={rgb, 255:red, 255; green, 255; blue, 255 }  ,fill opacity=1 ] (360,102.83) -- (390,102.83) -- (390,124.53) -- (360,124.53) -- cycle ;
\draw   (210,101.75) .. controls (210,97.08) and (207.67,94.75) .. (203,94.75) -- (159,94.75) .. controls (152.33,94.75) and (149,92.42) .. (149,87.75) .. controls (149,92.42) and (145.67,94.75) .. (139,94.75)(142,94.75) -- (97,94.75) .. controls (92.33,94.75) and (90,97.08) .. (90,101.75) ;
\draw   (390,101.75) .. controls (390,97.08) and (387.67,94.75) .. (383,94.75) -- (340.13,94.75) .. controls (333.46,94.75) and (330.13,92.42) .. (330.13,87.75) .. controls (330.13,92.42) and (326.8,94.75) .. (320.13,94.75)(323.13,94.75) -- (277,94.75) .. controls (272.33,94.75) and (270,97.08) .. (270,101.75) ;
\draw    (73,101.8) -- (73,62.7) ;
\draw    (226,101.75) -- (226,69.76) ;
\draw    (255,101.75) -- (255,69.76) ;
\draw    (405,101.75) -- (405,62.71) ;
\draw    (73,62.7) -- (214,62.7) ;
\draw    (264,62.71) -- (405,62.71) ;
\draw    (448,101.75) -- (448,62.71) ;
\draw    (405,62.71) -- (448,62.71) ;

\draw (69,126.61) node [anchor=north west][inner sep=0.75pt]  [font=\small]  {$2$};
\draw (100,126.61) node [anchor=north west][inner sep=0.75pt]  [font=\small]  {$3$};
\draw (182,127.78) node [anchor=north west][inner sep=0.75pt]  [font=\small]  {$\frac{m-1}{2}$};
\draw (212,127.78) node [anchor=north west][inner sep=0.75pt]  [font=\small]  {$\frac{m+1}{2}$};
\draw (242,127.78) node [anchor=north west][inner sep=0.75pt]  [font=\small]  {$\frac{m+3}{2}$};
\draw (272,127.78) node [anchor=north west][inner sep=0.75pt]  [font=\small]  {$\frac{m+5}{2}$};
\draw (360,128.6) node [anchor=north west][inner sep=0.75pt]  [font=\scriptsize]  {$m-1$};
\draw (399,128.61) node [anchor=north west][inner sep=0.75pt]  [font=\scriptsize]  {$m$};
\draw (139,128.4) node [anchor=north west][inner sep=0.75pt]  [font=\large]  {$\cdots $};
\draw (319,128.4) node [anchor=north west][inner sep=0.75pt]  [font=\large]  {$\cdots $};
\draw (138,74.9) node [anchor=north west][inner sep=0.75pt]   [align=left] {row};
\draw (306,71.89) node [anchor=north west][inner sep=0.75pt]   [align=left] {column};
\draw (220,53.51) node [anchor=north west][inner sep=0.75pt]   [align=left] {shared};
\draw (443,128.61) node [anchor=north west][inner sep=0.75pt]  [font=\scriptsize]  {$y$};
\draw (34,105.8) node [anchor=north west][inner sep=0.75pt]  [font=\large]  {$R$};

\end{tikzpicture}
    \caption{The partition of the randomness of \cref{alg:bghsv_query_pattern}.}
    \label{fig:sep_partition}
\end{figure}

\paragraph{Rectangular Neighbor-Listing (RNL)}

The BGHSV verifier makes $4 \cdot |\F|$ queries.
We index the queries by  $k := (b_1,b_2,t)\in \{0,1\}^2\times \F$ as follows.
\begin{itemize}
    \item $b_1 = 0$ indicates the query is to a line. $b_1 = 1$ indicates it is to a shifted line.
    \item $b_2 = 0$ indicates the query is to a first-axis parallel line. $b_2 = 1$ indicates it is to a canonical line.
    \item $t \in \F$ indicates the position on the line.
\end{itemize}

Given a configuration $(R, k)$ that results in a query to location $z\in \F^m$, we ought to show how to list all neighboring configurations (i.e. configurations that lead to location $z$) in a rectangular way.
We first show who are the neighboring configurations, and then show how to rectangularly (and synchronously) list them by describing the listing agents $A_{\row}$ and $A_{\column}$. This shows that the verifier has RNL.

\paragraph{Neighbors of $\pars{R, k}$.}

A full configuration $(R,k) = (R_\row, R_\column, R_\shared,k)$ specifies $k = (b_1, b_2, t)$, $x = (0, R_2, \ldots, R_m)$, and $R_y \in [|\bset_\lambda|]$. $R_y$ and $b_2$ determine $y \in \F^m$ as follows: if $b_2=0$ then $y=(1,0,\ldots, 0)$, otherwise $y = \bset_{\lambda}[R_y]$.

Recall that $\shift\pars{x}$ denotes cyclic shift of $x$ one step to the left. Given the full configuration specified by $k =  (b_1, b_2, t)$, $x$ and $R_y$, the location of the $k$th query will be $x+t\cdot y$ if $b_1=0$, or $\shift(x+t\cdot y)$ if $b_1=1$. More concisely, letting $\shift_j\pars{x}$ denote the cyclic shift of $x$ by $j$ steps for any $j \in \Z$, the location of the $k$th query is $\shift_{b_1}(x+t\cdot y)$.

Thus, any other configuration $(R', k')$ that specifies $k' = (b'_1, b'_2, t')$, $x'$ and $R'_y$ would query the same location if and only if
\begin{equation}\label{eq:who_are_the_neighbors}
\shift_{b'_1}(x'+t'\cdot y') = \shift_{b_1}(x+t\cdot y).
\end{equation}
In other words, $\pars{R',k'}$ neighbors $\pars{R,k}$ if and only if it satisfies \cref{eq:who_are_the_neighbors}.

\paragraph{The neighbor listing agents.}
Rearranging \cref{eq:who_are_the_neighbors}, we see that $(x, b_1, b'_1, t, t', y, y')$ uniquely determines $x'$ by the equation
\begin{equation}\label{eq:x'}
x' = \shift_{b_1-b'_1}(x+t\cdot y) - t'\cdot y',
\end{equation}

Recall that any location queried must fulfill the condition $x'_1 = 0$. For any possible $k' = \pars{b'_1, b'_2, t'}$ and $y'$, fulfillment of this condition is determined only by $k =\pars{b_1, b_2, t}$, $y$ and $\pars{x_2,x_m}$. Thus, we say that partial configuration $\pars{R'_{y}, k'}$ is \defemph{realizable (for $\pars{R_\shared,k}$)} if any only if $x'_1 = 0$.

Furthermore, notice that for any $(R'_{y},k')$ there is a unique $t'$ such that \cref{eq:x'} has $x'_1 = 0$, and finding such $t'$ can be done in a constant number of basic arithmetic operations. This will come in handy shortly, when we construct the listing agents.

We can now present two algorithms listing all neighbors of a given configuration in a rectangular and synchronized fashion. First the row neighbor-listing agent $A_\row$. On input row configuration $\pars{R_\row, R_\shared, k}$,
\begin{enumerate}
    \item \label{agents_obtain} Obtain $t$, $b_1$, $b_2$, $y$, $x_{[2,(m+3)/2]}$ and $x_m$ from the input. Initialize an empty list $L_\row$.
    \item For each $\pars{R'_y, b'_1, b'_2}$:
    \begin{enumerate}
        \item\label{agents_t'} Find $t' \in \F$ such that $\pars{R'_y, k'} = \pars{t', b'_1, b'_2}$ is realizable (as previously explained).
        \item\label{agents_x'} Compute $x'_\row \coloneqq x'_{[2, (m+1)/2]}$ as in \cref{eq:x'}. Append $\pars{x'_\row, R'_y, k'}$ to the list $L_\row$.
    \end{enumerate}
    \item \label{agents_sort} Sort $L_\row$ according to $\pars{R'_y, k'}$. Output $L_\row$.
\end{enumerate}

The column neighbor-listing agent runs similarly. On input column configuration $\pars{R_\column, R_\shared, k}$,
\begin{enumerate}
    \item Obtain $t$, $b_1$, $b_2$, $y$, $x_2$ and $x_{[(m+1)/2, m]}$ from the input. Initialize an empty list $L_\column$.
    \item For each $\pars{R'_y, b'_1, b'_2}$:
    \begin{enumerate}
        \item Find $t' \in \F$ such that $\pars{R'_y, k'} = \pars{t', b'_1, b'_2}$ is realizable (as previously explained).
        \item Compute $x'_\column := x'_{[(m+3)/2, m]}$ as in \cref{eq:x'}. Append $\pars{x'_\column, R'_y, k'}$ to the list $L_\column$.
    \end{enumerate}
    \item Sort $L_\column$ according to $\pars{R'_y, k'}$. Output $L_\column$.
\end{enumerate}

Note that both agents give lists indexed and sorted by all realizable $(R'_y, k')$. Thus, both arrays are of the same length and ordering. Moreover, the ordering of $L_{\row}$ and $L_{\column}$ when the agents are input any two neighboring configurations $\pars{R,k}$ and $\pars{R', k'}$. Thus, \cref{rnl_sync} of \cref{def:RNL} is satisfied.

For each $i \in \brackets{\abs{L_\row}}$, concatenating $L_{\row}[i]$ and $L_{\column}[i]$ gives a tuple $\pars{x'_\row, R'_y, k', x'_\column, R'_y, k'}$. By splitting $R'_y$ arbitrarily into $R'_{y.\row}$ and $R'_{y.\column}$ we have that $L_\row$ and $L_\column$, when appropriately ``zipped'' (as in \cref{rnl_zipped} of \cref{def:RNL}) give the list $L$ of all neighboring configurations to $\pars{R,k}$.

To see \cref{rnl_index} of \cref{def:RNL}, note that both $A_{\row}$ and $A_{\column}$ can identify the index of $(R,k)$ in $L$ since it corresponds to the index of the entry whose two last elements are $(R_y,k)$ in $L_{\row}$ and in $L_{\column}$.

Lastly, we calculate $t_\RNL$, which is the runtime of $A_\row$ (the case of $A_\column$ is analogous). \cref{agents_obtain} takes $\poly(m\cdot \log |\F|\cdot \log(1/s))$ time; this includes getting $y\in \bset_\lambda$ from a random string $R_y$ using the efficient construction of $\lambda$-biased set in \cite{AlonGHP1992,NaorN1993}. For each $(R'_y, b'_1, b'_2)$, \cref{agents_t',agents_x'} each take  $O(m\cdot \log |\F|)$ time. The output list length is upper bounded by $4|\bset_\lambda| = \poly(m\cdot \log |\F|\cdot \log(1/s))$, so sorting it in \cref{agents_sort} takes at most $\poly(m\cdot \log |\F|\cdot \log(1/s))$ time. Thus, overall the running time is dominated by $\poly(m\cdot \log |\F|\cdot \log(1/s))$. By the choice of parameters in \Cref{thm:outerBGHSV} and size of $\F$, this is asymptotically equal to $\poly(\log T(n))$.\end{proof}

\subsection{Alphabet Reduction}
\label{sec:alp_red}
\Cref{lemma:BGHSV_prop_large_alphabet} gives a robust PCP with RNL \emph{over a large alphabet}, but the final construct requires a PCP over the \emph{Boolean alphabet}. Next, we show that standard alphabet reduction \cite{Forney1965} preserves RNL. Namely, each symbol is replaced with its encoding in a binary error correcting code. We use the following lemma which gives a constant rate and constant distance code such that the decoding and encoding time is linear.

\begin{theorem}[{\cite{Spielman1996}}]
\label{lemma:ecc_linear}
There is a constant rate, constant distance linear error correcting code with a linear-time encoder, and a linear-time decoder recovering a message from a codeword with up to a fixed constant fraction of errors.
\end{theorem}

\begin{lemma}
\label{lemma:alphabet_red_RNL}

Suppose language $L$ has a PCP with verifier $V$ as described in \cref{fig:alphabet_reduction}, then $L$ has a PCP with verifier $V'$ as described in \cref{fig:alphabet_reduction}. Furthermore, if $V$ has $\tau$-RNL, then so does $V'$.

\begin{table}[ht]
    \centering
    \begin{tabular}{|c|c|c|}
        \hline 
        Complexity & $V$ & $V'$\tdelim
        \hline 
        Alphabet & $\Sigma$ & $\braces{0,1}$ \tdelim
        Robust soundness error & $s$ & $s$ \tdelim
        Robustness parameter & $\rho$ & $\Omega(\rho)$ \tdelim
        Randomness & $r$ & $r$ \tdelim
        Query & $q$ & $O(q\cdot \log |\Sigma|)$ \tdelim
        Proof length & $m$ & $O(m \cdot \log|\Sigma|)$ \tdelim
        Decision & $d$ & $d\cdot \polylog(|\Sigma|)$ \tdelim
        Runtime & $t$ & $t\cdot\polylog(|\Sigma|)$ \tdelim
        \end{tabular}
    \caption{The complexities of original verifier $V$ and the Boolean verifier $V'$.}
    \label{fig:alphabet_reduction}
\end{table}

\end{lemma}

\begin{proof}
	Fix a linear-time computable and decodable error correcting code of constant rate and distance, denoted $\Enc \colon \Sigma \rightarrow \{0,1\}^{\sigma}$ for $\sigma = O\pars{\log\abs{\Sigma}}$ from \Cref{lemma:ecc_linear}. 
	Assume that the first $\log\abs{\Sigma}$ bits in the encoding are the binary representation of the (non-binary) message. The new (Boolean) proof is written in a natural way: each non-binary symbol is replaced with its encoding under $\Enc$.
	
	The Boolean PCP verifier $V^\prime$ emulates the non-Boolean verifier $V$ as follows: when the non-Boolean verifier queries a location $b\in [m]$, the Boolean verifier queries the whole block $[(b-1)\sigma+1, b\sigma]$ and checks if it is a valid encoding of the symbol specified by the first $\log\abs{\Sigma}$ bits (if not, reject). Once all the queries are decoded correctly, the verifier $V'$ does the verification on the decoded values as $V$.
	
	It is easy to observe that the query complexity and the proof length increase by a factor of $O(\log |\Sigma|)$. Since the new verifier has to perform the decoding of an error correcting code, this adds a multiplicative overhead of $\polylog(|\Sigma|)$ in the running time and the decision complexity. 
	
	The importance of using the error correcting code is to make sure that the new verifier $V^\prime$ is still robust. This follows from \cite[Lemma 2.13]{BenSassonGHSV2006}, where it was shown that the soundness error remains the same and the robustness parameter decreases by a constant factor.
	
	It is also easy to observe that RNL is preserved with the same partition of the randomness. Fix a $j$-th location from the block  $[(b-1)\sigma+1, b\sigma]$. First observe the following proposition:
	\begin{description}\item{\bf Proposition (a):}\label{rop_proposition} If a full configuration $(R,k)$ queries a location $b\in [m]$ in the original proof, then the configuration $(R, \sigma (k-1)+j)$ queries the $j$-th location from the block $[(b-1)\sigma+1, b\sigma]$ in the new proof and vice-versa. 
	\end{description}
	Let $A_\row$ and $A_\column$ are the row and column agents of the non-Boolean verifier $V$. The row agent $A^\prime_\row(R_\row, R_\sharedrow, \tilde{k}) \rightarrow L^\prime_\row $ and the column agent $A^\prime_\column(R_\column, R_\sharedcolumn, \tilde{k}) \rightarrow L^\prime_\column $ of the new verifier $V^\prime$ are as follows: Both the agents first compute the block number $k = \lceil \tilde{k}/\sigma \rceil$ and the index 
	$j = \tilde{k} - \sigma(k-1)$.
	\begin{align*}
	    L^\prime_\row [i] \coloneqq  (R^\prime_\row, R^\prime_\sharedrow, \sigma(k' - 1) + j) \quad &s.t. \quad (R^\prime_\row, R^\prime_\sharedrow, k') = L_\row[i]\\
	    L^\prime_\column [i] \coloneqq  (R^\prime_\column, R^\prime_\sharedcolumn, \sigma(k' - 1) + j) \quad &s.t. \quad (R^\prime_\column, R^\prime_\sharedcolumn, k') = L_\column[i]
	\end{align*}
	where $L_\row \leftarrow A_\row(R_\row, R_\sharedrow,k)$ and $L_\column \leftarrow A_\column(R_\column, R_\sharedcolumn, k)$. To see the correctness, suppose the full configuration $(R_\row, R_\column, R_\sharedrow, R_\sharedcolumn, k)$ queries the location $b\in [m]$ from the original proof. By RNL of $V$, the full configuration $(R^\prime_\row, R^\prime_\column, R^\prime_\sharedrow,R^\prime_\sharedcolumn, k')$ given by $L_\row[i]$ and $L_\column[i]$ leads to the same location $b\in [m]$. Using Proposition (a), we can conclude that the verifier $V'$ when given the full configuration $(R^\prime_\row, R^\prime_\column, R^\prime_\sharedrow,R^\prime_\sharedcolumn, \sigma(k'-1)+j)$, queries the location $(b-1)\sigma+j$ in the new proof. As $\tilde{k} = \sigma(k-1) + j$, again using  Proposition (a), the input full configuration to the agents $(R^\prime_\row, R^\prime_\column, R^\prime_\sharedrow,R^\prime_\sharedcolumn, \tilde{k})$ also leads to the same location $(b-1)\sigma+j$ in the new proof. Therefore, every full configuration given by $(L^\prime_\row[i], L^\prime_\column[i])$ leads to the location $(b-1)\sigma+j$ in the new proof. Furthermore, since the number of full configurations on which $V$ queries $b\in [m]$ is the same as the number of full configurations on which $V'$ queries any fixed location from the block $[(b-1)\sigma+1, b\sigma]$ (in fact, there is a bijection given by  Proposition (a)), the list $(L^\prime_\row, L^\prime_\column)$ is exhaustive.
	\end{proof}

We now finish the proof of \Cref{thm:outerBGHSV}.
\begin{proof}[{Proof of \Cref{thm:outerBGHSV}}]
    \Cref{lemma:BGHSV_prop_large_alphabet} shows the existence of a verifier with the additional $\tau$-RNL over a large alphabet. The alphabet reduction technique from \Cref{lemma:alphabet_red_RNL} converts the PCP to a PCP over the Boolean alphabet.  Since, the original PCP is over an alphabet of size $2^{\polylog(T(n))}$, this conversion increases the proof length, query complexity, decision complexity and the verifier's running time by a multiplicative factor of $\polylog(T(n))$.  With all these changes, these four parameters of the new verifier are asymptotically same as the ones mentioned in \Cref{thm:outerBGHSV}.

    In the whole process, the robustness parameter of the verifier changes by a constant multiplicative factor and this change is irrelevant in proving the lemma.
\end{proof}

\section{Adding ROP to a robust PCP}\label{sec:rop}

One way at looking at the verification procedure is as follows: On sampling the randomness $R$, the verifier constructs a circuit $D := D(R)$ and a subset $I := I(R)$ of proof locations of size $q$. The verifier outputs the verdict of $D(\pi\restrict{I})$. In this abstract way, the circuit $D$ depends on the full randomness $R$. However, for our application we need the verifier to have randomness-oblivious circuits (ROP). 

Recall that the ROP states that the decision predicate depends only on a small fraction of the randomness, but may take as input a limited number of parity checks on the entire randomness. We generalize robust soundness to the ROP setting in the natural way, measuring the distance of both the bits read by the verifier \emph{as well as the randomness parity checks} from satisfying the decision predicate. This definition will be useful when we compose a robust PCP having ROP with a PCPP (in \Cref{section:composition}).

\begin{definition}[robust soundness for ROP verifier]
\label{def:robustnessROP}
For functions
$s,\rho : \pintegers\rightarrow [0,1]$,
a PCP verifier $V$ for a language $L$
with $\tau$-ROP and parity check complexity $p$
has \defemph{robust-soundness error} $s$ with \defemph{robustness parameter $\rho$}
if the following holds for every $x\notin L$:
For every oracle $\pi$,
the input to the decision predicate (that consists of bits read by the verifier and parities of the randomness) are $\rho$-close to being
accepted with probability
strictly less than $s$. Formally,
\[\forall \pi
\Pr_{(I,D,P)\getsr V(x)}
    [\mbox{$\exists a,b$ s.t. $D(ab)=1$ and $\delta(ab,\pi\restrict{I} P)\leq \rho$}]
  < s(|x|).\]
\end{definition}

\begin{lemma}
\label{lemma:adding_rop}
Suppose language $L$ has a PCP with verifier $V$ as described in \cref{fig:adding_rop}, then $L$ has a PCP with verifier $V'$ as described in \cref{fig:adding_rop} with $0$-ROP. Furthermore, if $V$ has $\tau$-RNL, then so does $V'$.
\begin{table}[ht]
    \centering
    \begin{tabular}{|c|c|c|}
        \hline 
        Complexity & $V$ & $V'$\tdelim
        \hline
        Robust soundness error & $s$ & $s$ \tdelim
        Robustness parameter & $\rho$ & $\Omega(\rho)$ \tdelim
        Randomness & $r$ & $r$ \tdelim
        Query & $q \geq r$ & $q$ \tdelim
        Proof length & $m$ & $m$ \tdelim
        Decision & $d$ & $\widetilde{O}\pars{t}$ \tdelim
        Parity-check & $-$ & $q$ \tdelim
        Runtime & $t$ & $\widetilde{O}\pars{t}$ \tdelim
        \end{tabular}
    \caption{The complexities of the original verifier $V$ and the $0$-ROP verifier $V'$.}
    \label{fig:adding_rop}
\end{table}
\end{lemma}

\begin{proof}

We can replace the circuit $D$ of $V$, which depends on the randomness $R$, with another circuit $D_\rop$ such that $D_\rop(x,R) = D(x)$ for all $x\in \{0,1\}^q$ and randomness $R$. In other words, we give $R$ as an explicit input to the circuit $D_\rop$ and therefore remove the dependence of $D_\rop$ on $R$. Note however that the output of $D_\rop$ depends on both the randomness $R$ as well as the proof locations $\pi\restrict{I}$.

It is easy to see that this preserves the completeness and soundness of the PCP. However, this transformation might lose the guarantee on the {\em robust soundness} when we look at the input to the circuit. This is because even if $\pi\restrict{I}$ is far from satisfying $D$, it might be the case that $(\pi\restrict{I}, R)$ is close to satisfying $D_\rop$ (in this case when measuring the distance from satisfying answers, changes in $R$ are also allowed).

In order to overcome this issue we encode the randomness with a good error code. By suitably repeating the code of \Cref{lemma:ecc_linear}, we have a linear error correcting code $\Enc \colon \bool^r \to \bool^q$ with constant relative distance.
Based on its linear time decoder, two circuits are constructed:

\begin{itemize}
    \item Circuit $\Dec$ that on input $y \in \bool^q$ outputs $x \in \bool^r$ such that $\Enc\pars{x} = y$ if such $x$ exists, and an arbitrary value otherwise.\footnote{In fact, the circuit only needs decode valid codewords, which is easier than decoding noisy codewords.}
    \item Circuit $\Test$ that on input $y \in \bool^q$ outputs $1$ if $y$ is a codeword (i.e., in the image of $\Enc$), and $0$ if it is not.
\end{itemize}

The actual decision circuit generated by $V'$ is $D'\pars{z, y} \coloneqq D_\rop\pars{z, \Dec\pars{y}} \wedge \Test\pars{y}$. The parity checks generated by $V'$ are simply $\Enc\pars{R}$ -- indeed, since the encoding is linear then $\Enc\pars{R}_1,\dots, \Enc\pars{R}_q$ are linear functions in $R$.

\paragraph{Decision and runtime complexities.}
The original verifier $V$ constructs $D$ in time $t$. The circuit $D_\rop$ can emulate this using $\widetilde{O}\pars{t}$ gates \cite{Cook1988}. By \cref{lemma:ecc_linear}, the runtime of decoding and testing is $O\pars{q}$, thus circuits $\Dec$ and $\Test$ can be constructed using at most $\widetilde{O}\pars{q} = \widetilde{O}\pars{t}$ gates. All in all, the size of the new decision circuit $D'$ is at most $\widetilde{O}\pars{t}$.

Apart from constructing $D'$, the verifier $V'$ also produces the queries of $V$ and the parity checks $\Enc\pars{R}$. Queries are constructed in time at most $t$, and again by \cref{lemma:ecc_linear} computing $\Enc\pars{R}$ takes time $O\pars{q}$. Thus the runtime of $V'$ is at most $\widetilde{O}\pars{t}$.

\paragraph{Robust soundness.}
The original verifier has robust-soundness error $s$ with robustness parameter $\rho$, so with probability at least $1-s$ the input $\pi\restrict{I}$ to the circuit $D:=D(R)$ is at least $\rho$-far from satisfying $D$. This means that with probability at least $1-s$, $(\pi\restrict{I}, E(R))$ is at least $\min\{ \rho/2, \Omega(1)\}$ far from satisfying $D'$. To see that, note that to satisfy $D'$ either the first half of the input $\pi\restrict{I}$ needs to be changed in $q\cdot \rho$ locations, or the second half of the input $E(R)$ needs to be changed to another legal encoding $E(R')$ which by the properties of $E$ requires $\Omega(q)$ bit-changes. Thus, the robustness parameter of $V'$ is $\min\{ \rho/2, \Omega(1)\}\geq \Omega(\rho)$, and the robust soundness error remains $s$.
\end{proof}


\section{RNL-preserving PCP composition}
\label{section:composition}
In this section, we strengthen the composition theorem of \cite{BenSassonGHSV2006} by showing that it preserves RNL and (to some extent) ROP of a robust PCP.

\subsection{PCPs of Proximity}

Recall that we think of a PCP verifier as accepting input $x$ and proof $\pi$ if the answers received from $\pi$, denoted $\pi\restrict{I}$, satisfy the decision circuit $D$ generated by $V$. In a nutshell, PCP composition is done by replacing the naive verification of the claim \emph{``$\pi\restrict{I}$ satisfies $D$''} with a verification by an \defemph{inner} verifier.

The goal of PCP composition is to reduce the query complexity of an (\defemph{outer}) verifier by composing it with an inner verifier of smaller (even constant) query complexity. Hence, the inner verifier has restricted access not only to its proof, but also to part of its input (namely, $\pi\restrict{I}$). Since such a constrained verifier cannot distinguish between answers that satisfy $D$ to those that are close to satisfying $D$, its soundness condition is relaxed to rejection of answers that are \emph{far} from all satisfying assignments to $D$. Indeed, if the outer PCP had suitable robustness, this relaxation still yields a sound PCP. We formalize this discussion next.

\begin{definition}[Pair language and \CVP]
    A \defemph{pair language} $L$ is a subset of $\bool^\ast \times \bool^\ast$. The \defemph{Circuit Valuation Problem} (\CVP) is the pair language consisting of circuits and their accepting inputs. Formally,
    \begin{equation*}
        \CVP \coloneqq \braces{\pars{C,y} \vert C\pars{y} = 1}.
    \end{equation*}
\end{definition}

A PCP of Proximity for a pair language $L$ is given a pair $\pars{x,y}$ and a proof $\pi$, where access to $x$ is \emph{explicit} (i.e., $x$ can be read entirely) while only oracle access is given to $y$ and $\pi$ (so queries to $y$ are accounted for in the verifier's query complexity). The soundness condition is weakened to rejection with high probability only of $\pars{x,y}$ such that $y$ is \emph{far} from $L\pars{x} \coloneqq \braces{y' \vert \pars{x,y'} \in L}$. Formally:
\begin{definition}[PCP of Proximity (PCPP)]
Let $L \subseteq \bool^\ast \times \bool^\ast$ be a pair language. For $s,\delta \colon \N \to \N$, a restricted verifier $V$ over $\Sigma$ is a \defemph{PCP of proximity verifier} for $L$ with \defemph{proximity parameter} $\delta$ and \defemph{soundness error} $s$ if the following two conditions hold for any $x,y \in \bool^\ast$:
\begin{description}
	\item[Completeness] If $(x,y)\in L$, then there exists a proof $\pi$ such that $V(x)$ accepts the oracle $y \pi$ ($y$ is called the {\em input oracle} and $\pi$ is called the {\em proof oracle}) with probability $1$. Formally, \[\exists \pi \qquad \Pr_{(I,D)\getsr V(x)}[D(y\pi\restrict{I}) = 1]=1.\]
	\item[Soundness] If $y$ is $\delta$-far from the set $L(x) = \{ y'\mid (x,y')\in L\}$, then for every  proof oracle $\pi$, $V(x)$ accepts the oracle $y \pi$ with probability strictly less than $s$. Formally, \[\forall \pi \qquad
\Pr_{(I,D)\getsr V(x)}[D(y\pi\restrict{I}) = 1] < s(|x|).\]
\end{description}
For convenience, we assume that the input locations $I = \pars{i_1, \dots, i_q}$ are each of the form $i_k = \pars{b_k, j_k}$, where $b_k$ is a bit signifying the oracle of the $k$-th query, and $j_k$ is the location in that oracle. Formally, for each $k \in \brackets{q}$, if $b_k = 0$ (resp. $b_k = 1$) then $j_k \in \brackets{\abs{y}}$ (resp. $j_k \in \brackets{\abs{\pi}}$) is the location in $y$ (resp. in $\pi$) of the $k$-th query of $V$.
\end{definition}
\subsection{The composition theorem}
Ben-Sasson \etal \cite{BenSassonGHSV2006} show that the composition of a robust PCP $V_\outer$ and a PCP of proximity $V_\inner$ with suitable parameters yields a sound \defemph{composite} PCP (described in \cref{fig:composition_bghsv}) that enjoys the inner query complexity and (roughly) the outer randomness complexity.

\begin{figure}[ht]
	\begin{tcolorbox}[
		standard jigsaw,
		opacityback=0.5, 
		]
\begin{description}
    \item[Hardwired:] Outer verifier $V_\outer$ and inner verifier $V_\inner$.
    \item[Input:] Explicit input $x$, outer proof $\Pi$ and inner proofs $\braces{\pi_{R_\outer} \;\vert\; R_\outer \in \bool^{r_\outer}}$.
\end{description}
\begin{enumerate}

\item Sample $R_{\outer} \in \{0,1\}^{r_{\outer}}$.

\item Run $V_{\outer}(x;R_{\outer})$ to obtain $I_{\outer} = (i_1, i_2, \ldots, i_{q_{\outer}})$ and $D_\outer$.

\item Sample $R_{\inner}\in \{0,1\}^{r_{\inner}}$.

\item 
Run $V_{\inner}(D_{\outer}, R_{\inner})$ to obtain $I_{\inner} =  ((b_1,j_1),\ldots ,(b_{q_{\inner}},j_{q_{\inner}}))$ and $D_{\inner}$.

\item Initialize $I_\comp \coloneqq \emptyset$, and $D_\comp \coloneqq D_\inner$
\item For each $k \in \brackets{q_\inner}$, determine the queries of the composite verifier:
\begin{enumerate}
\item If $b_{k} = 0$, let $\compr{i}_k$ be the location of $\Pi(i_{j_k})$. Append $\compr{i}_k$ to $ I_\comp$.
\item If $b_k = 1$, let $\compr{i}_k$ be the location of $\pi_{R_{\outer}}(j_k)$. Append $\compr{i}_k$ to $ I_\comp$.
\end{enumerate}
\end{enumerate}
\begin{description}
    \item[Output] $(I_\comp, D_\comp)$.
\end{description}
\end{tcolorbox}
\caption{The composite verifier $V_\comp$ of \cite{BenSassonGHSV2006}.}
\label{fig:composition_bghsv}
\end{figure}

\begin{theorem}[{\cite[Theorem 2.7]{BenSassonGHSV2006}}]\label{thm:bghsv_composition}
    Suppose language $L$ has a robust PCP verifier $V_\outer$, and that $\CVP$ has a PCPP verifier $V_\inner$, with parameters described in \cref{fig:composition_parameters} such that $\rho_\outer \geq \delta_\inner$.
    Then, language $L$ has a PCP verifier $V_\comp$ as described in \cref{fig:composition_parameters}.
\end{theorem}

\begin{table}[ht]
    \centering
    \begin{tabular}{|c|c|c|c|}
        \hline 
        Complexity & $V_\outer$ & $V_\inner$ & $V_\comp$ \tdelim
        \hline 
        Soundness error & (Robust:) $s_\outer$ & $s_\inner$ & $s_\outer + s_\inner$ \tdelim
        Proximity parameter & - & $\delta_\inner$ & - \tdelim
        Robustness parameter & $\rho_\outer$ & - & - \tdelim
        Randomness & $r_\outer$ & $r_\inner$ & $r_\outer + r_\inner $\tdelim
        Query & $q_\outer$ & $q_\inner$ & $q_\inner$\tdelim
        Decision & $d_\outer$ & $d_\inner$ & $d_\inner$ \tdelim
        Runtime & $t_\outer$ & $t_\inner$ & $t_\outer + t_\inner$ \tdelim
        \end{tabular}
    \caption{The complexities of $V_\outer$, $V_\inner$ and $V_\comp$. The complexities of each verifier are taken with respect to its input; that is, the complexities of the outer and composite verifier are with respect to $n$, while those of the inner verifier are with respect to $d_\outer\pars{n}$. For example, $r_\outer + r_\inner$ refers to $r_\outer(n) + r_\inner(d_\outer(n))$.}
    \label{fig:composition_parameters}
\end{table}

Our goal is to reduce the query complexity of the PCP of \cref{section:outer} by composing it with a constant-query inner PCPP (see \cref{section:all_together}). Towards this end, we adapt the composition to consider ROP (described in \cref{fig:composition_bghsv_rop}), and strengthen \cref{thm:bghsv_composition} by showing that it preserves RNL of the outer PCP, and (to some extent) its ROP.
\begin{figure}[ht]
	\begin{tcolorbox}[
		standard jigsaw,
		opacityback=0.5, 
		]
\begin{description}
    \item[Hardwired:] Outer verifier $V_\outer$ and inner verifier $V_\inner$.
    \item[Input:] Explicit input $x$, outer proof $\Pi$ and inner proofs $\braces{\pi_{R_\outer} \vert R_\outer \in \bool^{r_\outer}}$.
\end{description}
\begin{enumerate}

\item Obtain outer verifier circuit $D_{\outer}$ from $V_\outer(x)$.\footnote{We use $0$-ROP of the outer verifier to obtain the circuit before sampling the outer randomness.}

\item Sample $R_{\outer} \in \{0,1\}^{r_{\outer}}$.

\item Run $V_{\outer}(x;R_{\outer})$ to obtain $I_{\outer} = (i_1, i_2, \ldots, i_{q_{\outer}})$ and the parity-checks $P_{\outer} = (C_1, \ldots, C_{p_{\outer}})$.

\item Sample $R_{\inner}\in \{0,1\}^{r_{\inner}}$.

\item 
Run $V_{\inner}(D_{\outer}, R_{\inner})$ to obtain $I_{\inner} =  ((b_1,j_1),\ldots ,(b_{q_{\inner}},j_{q_{\inner}}))$ and $D_{\inner}$.

\item Initialize $I_\comp, P_\comp \coloneqq \emptyset$, and $D_\comp \coloneqq D_\inner$
\item For each $k \in \brackets{q_\inner}$, determine the queries of the composite verifier:
\begin{enumerate}
\item If $b_{k} = 0$ and $j_k \leq i_{q_\outer}$, let $\compr{i}_k$ be the location of $\Pi(i_{j_k})$. Append $\compr{i}_k$ to $I_\comp$.
\item \label{step:parity}If $b_{k} = 0$ and $j_k > i_{q_\outer}$. Append $C_{j_k - q_\outer}$ to $P_\comp$.
\item If $b_k = 1$, let $\compr{i}_k$ be the location of $\pi_{R_{\outer}}(j_k)$. Append $\compr{i}_k$ to $I_\comp$.
\end{enumerate}
\end{enumerate}
\begin{description}
    \item[Output] $(I_\comp, P_\comp, D_\comp)$.
\end{description}
\end{tcolorbox}
\caption{The composite verifier $V_\comp$ of \cref{fig:composition_bghsv}, adapted to preserve ROP.}
\label{fig:composition_bghsv_rop}
\end{figure}

\begin{lemma}
\label{lemma:composition_RNL_rop}
    In the setting of \cref{thm:bghsv_composition}, assume further that $V_\outer$ has $0$-ROP with parity-check complexity $p_\outer$, and $\tau$-RNL with listing agent runtime of $t_\RNL$. Then $V_\comp$ has $\compr{\tau}$-ROP with parity-check complexity $p_\outer$, and $\compr{\tau}$-RNL with listing agent runtime of $t_\RNL \cdot 2^{r_\inner} \cdot q_\inner \cdot t_\inner$, where
    \begin{equation*}
        \compr{\tau} = \frac{r_\inner + \tau \cdot r_\outer}{r_\inner + r_\outer}.
    \end{equation*}
    Furthermore, the shared and aware parts of the randomness are the same.
\end{lemma}
\begin{remark}
    We stress that the outer verifier's ROP (as well as RNL) is used in showing RNL of the composite verifier. Briefly, this is so that the composite listing agents can emulate the inner verifier so as to obtain its queries, without having access to the outer randomness.
\end{remark}

\begin{proof}
First, some names and notation. In the following proof, $V_\outer$, $V_\inner$ and $V_\comp$ will be called the \defemph{outer}, \defemph{inner} and \defemph{composite} verifiers (respectively). We affix the terms \defemph{outer}, \defemph{inner} and \defemph{composite} when discussing components of the respective verifiers. For example, the outer randomness refers to the randomness used by the outer verifier. We denote the outer randomness by $R_{\outer}$, the inner randomness by $R_\inner$, and the composite by $\compr{R}$ -- this will avoid double-indices when we further partition the outer and composite randomness.

\paragraph{Soundness}
Note that the outer verifier $V_\outer$ on randomness $R_\outer$ queries the outer proof $\Pi$ at $q_\outer$ many locations $I$. The verifier then computes $p_\outer$ many parities $P_\outer\coloneqq (C_1(R_\outer), C_2(R_\outer), \ldots, C_{p_\outer}(R_\outer))$ of $R_\outer$ and evaluates the circuit $D_\outer(\Pi_{|I}, P_\outer)$. By the robust soundness property of $V_\outer$ with probability at least $(1-s_\outer)$ over $R_\outer$, the string $(\Pi_{|I_\outer}, P_\outer) $ is at least $\rho_\outer$-far from satisfying $D_\outer$. Therefore, with probability at least $(1-s_\outer)$, the inner verifier $V_\inner$ gets $(D_\outer, (\Pi_{|{I_\outer}}, P_\outer))$ instance of circuit value problem where $(\Pi_{|I_\outer}, P)$ is $\rho_\outer$-far\footnote{The guarantee on the input to the circuit is strong than just saying that it is $\rho_\outer$-far from satisfying $D_\outer$. This is because the inner verifier always gets the {\em honest} parities $C_i(R_\outer)$ when needed, unlike the proof queries which can be arbitrary in soundness case (See \Cref{step:parity} in \Cref{fig:composition_bghsv_rop}). However, we do not need this structure in proving the lemma. } from satisfying $D_\outer$. Since, the proximity parameter $\delta_\inner$ of $V_\inner$  satisfies $\delta_\inner \leq \rho_\outer$, $V_\inner$ rejects any such proofs with probability at least $1-s_\inner$. Therefore, the composite verifier rejects the proof with probability at least $(1-s_\outer)(1-s_\inner)$ and hence the soundness error is at most $1-(1-s_\outer)(1-s_\inner) \leq s_\outer + s_\inner$.

\paragraph{RNL}
To show RNL of $V_\comp$, we show that there is a partition of the composite randomness, as well as a row neighbor-listing agent $\compr{A}_\row$ and column neighbor-listing agent $\compr{A}_\column$.

Let the RNL partition of the outer randomness $R_{\outer}$ be given by
\begin{equation*}
    R = (R_\row, R_\column, R_\sharedrow, R_\sharedcolumn)\in \{0,1\}^{r_{\outer}}.
\end{equation*}

The partition of the composite randomness $\compr{R}$ (which consists of the randomness of both outer and inner verifier) is rather simple: the composite row and column parts are the same as the outer's, and the shared part is formed of the outer's shared part as well as the entire inner randomness. Formally, composite randomness $\compr{R}$ is partitioned into
\begin{align*}
    \compr{R}_\row &\coloneqq R_\row \\
    \compr{R}_\column &\coloneqq R_\column \\
    \compr{R}_\sharedrow &\coloneqq \pars{R_\sharedrow, R_\inner}\\
    \compr{R}_\sharedcolumn &\coloneqq \pars{R_\sharedcolumn, R_\inner}.
\end{align*}
As before, some bits of shared randomness (in this case, $R_{\inner}$) appear in both $R_\sharedrow$ and $R_\sharedcolumn$. Again, we can precisely meet the requirement in \Cref{def:RNL} by further partitioning the inner randomness into two parts. Then, indeed we have
\begin{equation*}
    \compr{\tau} \coloneqq \frac{ \abs{R_\inner} + \abs{R_\sharedrow} + \abs{R_\sharedcolumn} }{\abs{R_\inner} + \abs{R}} = \frac{r_\inner + \tau \cdot r_\outer}{r_\inner + r_\outer}.
\end{equation*}

We construct the row and column agents of $V_\comp$, denoted $\compr{A}_\row$ and $\compr{A}_\column$ (respectively). Fix a configuration $(\compr{R},k) \coloneqq (R_\outer, R_\inner, k) $  where the outer randomness is $R_\outer$ and the inner randomness is $R_\inner$. We describe only the row agent $\compr{A}_\row$, as the column agent $\compr{A}_\column$ is analogous. The outcome depends on whether the query was issued to the outer proof or to one of the inner proofs. (In the notation of \cref{fig:composition_bghsv_rop}, this is determined by the value of $b_k$.)

\begin{description}
\item[Case 1] The $k$-th query is to an inner proof (i.e., $b_k=1$).

Observe that if a configuration $(R'_\outer, R'_{\inner}, k')$ (\footnote{We stress that $R'_{\outer}$ and, $R'_{\inner}$ are the outer and inner randomness of the configuration.}) results in a query to the same location, then it must have the same outer randomness as $\pars{\compr{R},k}$; that is, $R'_\outer=R_\outer$. If indeed $R_\outer=R'_\outer$, checking whether the location queried by the two configurations is the same depends only on the inner randomness and the query index, and these are known given the composite shared randomness. In such a case, $\compr{A}_{\row}$ lists all neighboring configurations by finding all possible $(R'_{\inner}, k')$ that lead to the same location via ``brute-force'' enumeration. Thus, in this case, on input $\pars{R_\row, R_\shared, R_\inner}$, the row agent runs as follows:
\begin{enumerate}
    \item Initialize a list $\compr{L}_\row$.
    \item For each $R'_\inner \in \bool^{r_\inner}$ and $k' \in \brackets{q_\inner}$ such that $b_{k'}=1$:
    \begin{enumerate}
        \item Check if the inner configuration $\pars{R'_\inner,k'}$ results in the same query location as the configuration $\pars{R_\inner,k}$. If so, add $\pars{R_\row, R_\sharedrow, R'_\inner, k'}$ to the list $\compr{L}_\row$.
    \end{enumerate}
    \item Output $\compr{L}_\row$.
\end{enumerate}

\item[Case 2]  The $k$-th query is to the outer proof (i.e., $b_k=0$).

Let $(R'_{\outer},R'_{\inner},k')$ be a configuration that potentially leads to read the same proof location as $(R_{\outer}, R_{\inner}, k)$.
Denote by $(b'_{k'}, j'_{k'})$ the $k'$-query of the inner verifier on $(D_{\outer}, R'_{\inner})$.

In this case, we have that $(R_{\outer}, R_{\inner}, k)$ and $(R'_{\outer},R'_{\inner},k')$ are neighboring configurations if and only if
$b'_{k'}=0$ and the underlying configurations to $V_{\outer}$, namely $(R_{\outer},j_{k})$ and $(R'_{\outer}, j'_{k'})$, are neighboring configurations with respect to $V_{\outer}$.

Thus, we start by listing all neighboring configurations of $(R_{\outer}, j_k)$ with respect to $V_{\outer}$ in a rectangular and synchronized fashion. Then, we extend each such neighbor $(R'_{\outer}, j')$ to a list of all configurations $(R'_{\outer}, R'_{\inner}, k')$ to $V_{\comp}$ that read the same location in the outer proof.

Thus, in this case, on input $\pars{R_\row, R_\shared, R_\inner}$, the row agent runs as follows:

\begin{enumerate}
    \item Initialize a list $\compr{L}_\row$.
    \item Compute $j_k$ based on $D_\outer$ and $R_\inner$.
    \item Invoke the outer row agent $A_\row\pars{R_\row, R_\shared, j_k}$ to obtain a list $L_\row$.
    \item \label{comp_RNL_dominating_runtime} For each $(R'_{\row}, R'_{\sharedrow} ,j')$ in $L_\row$:
    \begin{enumerate}
        \item For each $R'_\inner \in \bool^{r_\inner}$ and $k' \in \brackets{q_\inner}$:
        \begin{enumerate}
            \item Emulate $V_\inner$ on input $D_\outer$ and randomness $R'_\inner$. If the $k^\prime$-th query is to location $j'$ in the input oracle, i.e., $b_{k'} = 0$ and $j_{k'} = j'$, add $\pars{R'_\row, R'_\sharedrow, R'_\inner, k'}$ to $\compr{L}_\row$.
        \end{enumerate}
    \end{enumerate}
    \item Output $\compr{L}_\row$.
\end{enumerate}

\end{description}
It is simple to verify that these lists are synchronized and canonical, and that the zip $\compr{L}$ of the two  lists $\compr{L}_{\row}$ and $\compr{L}_{\column}$ lists all neighboring configurations of $(\compr{R},k)$.

We show that $\compr{A}_{\row}$ knows the index of the full given configuration $(\compr{R}, k)$ in $\compr{L}_{\row}$ (an analogous statement holds for $\compr{A}_{\column}$). In case 2, observe that running $A_{\row}$ on $(R_{\row}, R_{\shared}, j_k)$ gives the unique index of the entry corresponding to $(R_{\outer},j_k)$ in $L_\row$. Furthermore, from the shared randomness, $\compr{A}_{\row}$ knows $(R_{\inner},k)$. Together this identifies uniquely the list entry in $\compr{L}_{\row}$ that corresponds to $(R, k)$. In case 1, this is even simpler, since from the shared randomness $\compr{A}_{\row}$ knows $(R_{\inner},k)$ and can thus identify the unique corresponding list entry in $\compr{L}_{\row}$.

The running time of both agents is dominated by the double enumeration in \cref{comp_RNL_dominating_runtime} of Case 2, which is dominated by
\begin{equation}
    t_\RNL \cdot 2^{r_\inner} \cdot q_\inner \cdot t_\inner.
\end{equation}

\paragraph{ROP}
We prove that the composite verifier has $\compr{\tau}$-ROP, where the aware and shared part of the randomness are the same. Since the outer verifier has $0$-ROP, then the description of $D_{\outer}$ is indeed independent of the randomness. Also, the set of queries $I_{\inner} =  ((b_1,j_1),\ldots ,(b_{q_{\inner}},j_{q_{\inner}}))$ and $D_{\inner}$ depend only on $R_{\inner}$, which is a part of the aware randomness. Thus, once the aware randomness is fixed, the query corresponding to  $(b_c, j_c)$ where $b_{c} = 0$ and $j_c > i_{q_\outer}$, is a fixed parity of the outer randomness $R_{\outer}$. Since the aware part of randomness contains the shared outer randomness, the $j_c$-th input to $D_\inner$ is a fixed affine (rather than linear) function of $R_\oblivious$.

This shows that the predicate of the composite verifier depends only on the aware randomness, which is the same as the shared randomness. Additionally, the randomness parity checks fed into the predicate are determined by the aware part of the randomness.
\end{proof}


\section{Obtaining the final construct}\label{section:all_together}
In this section, we obtain our final {\em short}, {\em efficient} and {\em constant-query} PCP for languages in $\NTIME(T(n))$ that is smooth, rectangular and has ROP.

We will use the randomness-efficient, constant-query PCPP of Mie in the inner level of our composition.

\newcommand{\mie}{\inner} 
\begin{theorem}[\cite{Mie2009}]
\label{thm:miepcpp}
Suppose $L$ is a pair language in $\NTIME(T(n))$ for some non-decreasing function $T(n)$. Then, for every constant $0<s,\delta<1$, $L$ has a PCPP verifier with the following parameters:
\begin{itemize}
\item randomness complexity $r_\mie(n) \coloneqq \log T(n) + O(\log \log T(n))$,
\item query complexity $q_\mie(n) \coloneqq O_{s, \delta}(1)$,
\item verification time $t_\mie(n) \coloneqq \poly(n,  \log K, \log T(n+ K))$.
\item soundness error $s$ and proximity parameter $\delta$.
\end{itemize}
Here $K$ is the length of the implicit input $y$.
\end{theorem}

We can now prove the following main theorem.
\begin{theorem}
\label{thm:combinedPCP}
Let $L$ be a language in
$\NTIME(T(n))$ for some non-decreasing function $T \colon \N \rightarrow \N$. There exists a universal constant $c$ such that for all odd integers $m \in \mathbb{N}$, and any constant $s \in (0,\frac{1}{2})$ satisfying $T(n)^{1/m} \geq m^{cm}/s^6$, $L$ has a PCP verifier with the following parameters:
\begin{itemize}
    \item Alphabet $\bool$.
    \item Randomness complexity $r\pars{n} = \log T(n) + O(m\log \log T(n))+ O(\log n)$.
    \item Soundness error $s$.
    \item Decision,  query and parity-check complexities all $O_s(1)$.
    \item Verifier runtime $t(n) = \poly(n,T(n)^{1/m})$.
    \item The verifier has $\tau$-ROP and is $\tau$-rectangular, where
    $$\tau \cdot r(n) = \frac{6}{m}\log T(n) +  O(m\log \log T(n)) + O(\log n).$$
    Furthermore, the shared and the aware parts of the randomness are the same.
\end{itemize}
\end{theorem}
\begin{proof}
We start with the robust PCP verifier over the Boolean alphabet for $L$ given by \Cref{thm:outerBGHSV}, with the following complexities:
\begin{itemize}
    \item Randomness complexity $r_\outer(n) = (1-\frac{1}{m})\log T(n) + O(m \log \log T(n)) + O(\log(1/s)).$
    \item Query and Decision complexity $q_\outer(n) = d_\outer(n) = T(n)^{1/m} \cdot \poly(\log T(n), 1/s).$
    \item Verifier running time $t_\outer(n) = q_\outer(n) \cdot \poly(n, \log T(n)).$
    \item $\tau_\outer$-RNL, with 
    $\tau_\outer \cdot r_\outer(n)= \frac{4}{m}\log T(n) + O(m\log \log T(n)) + O(\log(1/s))$,
    and listing-agents runtime $t_{\RNL, \outer}(n) = \poly(\log T(n))$.
    \item Robust soundness error $s/3$ with robustness parameter $\rho = \Theta(s)$.
    
\end{itemize}

We next add the $\tau_\outer$-ROP to this PCP verifier using \Cref{lemma:adding_rop}. This step increases the decision complexity and verifier's running time to $\widetilde{O}(t_\outer(n))$ and adds the parity-check complexity $q_\outer(n)$. It also reduces the robustness parameter by a constant factor to $\Omega(\rho)$.  

Next, we wish to compose this (outer) robust PCP with the inner PCPP of \cref{thm:miepcpp}. Since the decision complexity of the outer verifier is $\widetilde{O}(t_\outer(n))$ the inner PCPP ought to verify the pair-language {\CVP } $\in \NTIME(\widetilde{O}(t_\outer(n)))$. Therefore, we compose this robust PCP with the PCPP verifier of Mie given in \Cref{thm:miepcpp}, for the pair language {\CVP } $\in \NTIME(\widetilde{O}(t_\outer(n)))$,  with soundness error $s/3$ and proximity parameter which is greater than the robustness parameter $\Omega(\rho)$ of the outer verifier.  The PCPP has query complexity $O_s(1)$, randomness complexity $r_\mie(\widetilde{O}(t_\outer(n))) = \frac{1}{m}\log T(n) + O(\log \log T(n))+ O(\log n) + O(\log(1/s))$  and verification time $t_\mie(\widetilde{O}(t_\outer(n))) = \poly( t_\outer(n)) = \poly(n, T(n)^{1/m}, 1/s)$.

By \Cref{thm:bghsv_composition,lemma:composition_RNL_rop}, the composite PCP verifier has soundness error $2s/3$, and query and parity-check complexities both $O_s(1)$. The randomness complexity of the composite verifier, denoted by $r(n)$, is $r_\outer(n) + r_\mie(\widetilde{O}(t_\outer(n)))$. 
The running time of the composite verifier, denoted $t_{\mathrm{comp}}(n)$, is $\widetilde{O}\pars{t_\outer(n)} + t_\mie(\widetilde{O}\pars{t_\outer(n)}) = \poly(t_\outer(n)) =  \poly(n,T(n)^{1/m})$. 

 The composite verifier has $\compr{\tau}$-RNL and $\compr{\tau}$-ROP where
 $$\compr{\tau} = \frac{r_\mie(\widetilde{O}(t_\outer(n))) + \tau_\outer \cdot r_\outer(n)}{r_\mie(\widetilde{O}(t_\outer(n))) + r_\outer(n)} \leq \tau .$$
 Furthermore, the $\shared$ and the $\aware$ parts of the randomness are indeed the same. The running time of the RNL agents, denoted $t_{\RNL}(n)$, is at most 
 $$O(t_{\RNL, \outer}(n) \cdot 2^{r_\mie(\widetilde{O}(t_\outer(n)))} \cdot t_\mie(\widetilde{O}(t_\outer(n))) \leq  \poly(n,T(n)^{1/m})\;.$$

Finally, we use \Cref{thm:smoothification} with $\mu = s/3$ to make the composite verifier smooth and rectangular. This step increases the soundness error by $s/3$ and thus the overall soundness error is $s$ as required. The PCP verifier becomes smooth and is $\tau$-rectangular and has $\tau$-ROP, with the $\shared$ and the $\aware$ parts of the randomness still the same. This conversion keeps the randomness and parity-check complexities the same.
The decision and query complexities remain the same up to constants. 
Finally, the running time of the verifier is
$t(n) = q \cdot \poly(t_{\RNL}(n)) + t_{\comp}(n) =\poly(n,T(n)^{1/m})$,
as asserted. 
\end{proof}
\begin{remark}
We remark that the randomness complexity of the composed verifier in \cref{thm:combinedPCP} is actually only $\log T(n) + O(m\log \log T(n))$ for $T(n) = \Omega(n^m)$. This improvement is obtained by the improved verifier running time $t_\outer(n)=q(n)\cdot \polylog T(n) + O(n) = T(n)^{1/m} \cdot \poly(\log T(n)) + O(n)$ mentioned in \cref{rem:CLW}. Note that when $T(n) = \Omega(n^m)$, $t_\outer(n) = T(n)^{1/m} \cdot \poly(\log T(n))$. Plugging this value of $t_\outer(n)$ in the above proof yields the above randomness complexity. As in the case of \cref{rem:CLW}, this improvement is not needed for our construction.
\end{remark}


\section*{Acknowledgements}
We thank Lijie Chen and Emanuele Viola for their helpful comments on \cref{sec:related}, and an anonymous reviewer for pointing out \cref{remark:BarakGoldreich}.

{\small 
\bibliographystyle{prahladhurl}
\bibliography{rectangular-bib.bib}
}

\appendix


\section{Soundness and smoothness of \Cref{alg:new_verifier}} \label{smoothification_appendix}

With rectangularity out of the way, we can describe the run of the new verifier given proof $\newpi$ a more succinct way:

\begin{algorithm}[\cref{alg:new_verifier}, simplified]\label{alg:new_verifier_simplified}
Given input $x$ and proof $\pi_\new$, the new verifier $\newV$ runs as follows:
\begin{enumerate}
    \item Sample $R \in \bool^r$.
    \item \label{step:consistency_check}For each $k \in \brackets{q}$, query $\newpi$ for $\pars{R,k}$ as well as its sampler-neighbors. 
    \item 
    	Feed the $q\cdot \agla$ bits queried from the proof $\pi_{\new}$ to a circuit that first checks consistency between every sampler-neighborhood. That is, it check that in each of the $q$ blocks of $\agla$ bits, all the $\agla$ bits are equal. If an inconsistency is spotted the circuit immediately rejects. Otherwise, feed the first bit in every block to the decision circuit of the original verifier $\oldV$ (along with the $p$ parity-checks on the randomness) and output its answer.
\end{enumerate}
\end{algorithm}

Indeed, \cref{alg:new_verifier} and \cref{alg:new_verifier_simplified} describe the same verifier precisely due to RNL: for any random coin sequence $R$ and query index $k$, \Cref{find_neighbors} of \cref{alg:new_verifier} indeed queries all neighbors of $\pars{R, k}$ in the sampler (and then checks their consistency). We now show that this verifier is sound and smooth.
\newcommand{\overbar}[1]{\mkern 1.5mu\overline{\mkern-1.5mu#1\mkern-1.5mu}\mkern 1.5mu}

\newcommand{\HighConsistency}{H}
\newcommand{\LowConsistency}{\overbar{\HighConsistency}}
\newcommand{\Consistent}{M}

\subsection{Soundness}\label{sec:smooth_soundness}
Let $\alpha \coloneqq \mu/2q$. Recall that $\newV$ denotes the new smooth verifier, and $\oldV$ denotes the original verifier. Fix an input $x \notin L$ and an alleged proof $\newpi$ for $\newV$. We will show that $\newV$ rejects $x$ and $\newpi$ with probability at least $1 - s - \mu$.

We let $i^{(k)}\pars{R}$ denote the location of the $k$-th query of $\oldV$ when sampling random coins $R$. Recall that locations in $\newpi$ are indexed by full configurations $(R,k)$, that can be partitioned into $m$ disjoint $\alpha$-samplers, where the $j$-th sampler connects all configurations $(R,k)$ such that $i^{(k)}\pars{R} = j$. We derive a proof $\oldpi$ for the original verifier $\oldV$ by assigning $\oldpi\pars{j}$ the majority value of $\newpi$ on the $j$-th sampler. Formally,
\begin{equation*}
    \oldpi\pars{j} \coloneqq \Maj_{R,k}\braces{ \newpi\pars{R,k} \;\big\vert\; i^{(k)}\pars{R} = j}.
\end{equation*}

We say that $j \in \brackets{m}$ has \defemph{consistency} $1-2\alpha$ if $\pars{1-2\alpha}$-fraction of the configurations $\pars{R,k}$ pointing to it satisfy $\newpi\pars{R,k} = \oldpi\pars{j}$.

Two events in the verifier's randomness $R$ will be of particular interest:
\begin{itemize}
    \item $\HighConsistency$: for all $k \in \brackets{q}$, $i^{(k)}\pars{R}$ has consistency at least $1-2\alpha$.
    \item $\Consistent$: for all $k \in \brackets{q}$, $\newpi\pars{R,k} = \oldpi\pars{i^{(k)}\pars{R}}$.
\end{itemize}
Let $\delta = \Pr\brackets{\LowConsistency}$. We show the following two claims.
\begin{claim}\label{eq:dh_soundness_1}
    \begin{equation*}
        \Pr\brackets{\newV^\newpi = 0 \;\big\vert\; \LowConsistency} \geq 1 - q\alpha.
    \end{equation*}
\end{claim}
\begin{claim}\label{eq:dh_soundness_2}
    \begin{equation*}
        \Pr\brackets{\newV^\newpi = 0 \;\big\vert\; \HighConsistency} \geq 1 - \frac{s}{1 - \delta} - 2 q \alpha.
    \end{equation*}
\end{claim}

These two claims will give the required result, because
\begin{align*}
    \Pr\brackets{\newV^\newpi = 0} &= \Pr\brackets{\LowConsistency} \cdot  \Pr\brackets{\newV^\newpi = 0 \;\big\vert\; \LowConsistency} + \Pr\brackets{\HighConsistency} \cdot \Pr\brackets{\newV^\newpi = 0 \;\big\vert\; \HighConsistency} \\
    &\geq \delta\pars{1 - q\alpha} + \pars{1 - \delta} \pars{1 - \frac{s}{1 - \delta} - 2 q \alpha} \\
    &=\delta - \delta q\alpha + (1-\delta) - s - (1-\delta)2q\alpha\\
    &\geq 1-s - 2q\alpha = 1 - s - \mu.
\end{align*}

Now let's prove the claims.

\begin{proof}[Proof of \cref{eq:dh_soundness_1}]
    Consider a sampler with consistency at most $1-2\alpha$. The average value\footnote{Recall that $\newpi$ is over $\{0,1\}$ and we interpret these as real numbers.} of $\newpi$ on this sampler is between $2\alpha$ and $1 - 2\alpha$. Since this is a $\alpha$-sampler, it holds that for all but $\alpha$-fraction of configurations (named \defemph{error configurations}), the average value of $\newpi$ on the sampler-neighborhood of each configuration is between $\alpha$ and $1 - \alpha$. In particular, $\newpi$ assigns inconsistent values to the sampler-neighbors of each non-error configuration.
    
    Denoting the set of all error configurations by $E$, we see that sampling an $R \in \overbar{H}$ such that $\pars{R,k} \notin E$ for all $k$ leads the $\newV$ to reject in the consistency check, \cref{step:consistency_check} of \cref{alg:new_verifier_simplified}. We upper bound the probability of the complement event,
    \begin{equation*}
        \Pr_{R}\brackets{\exists k \in \brackets{q} \quad  \pars{R,k} \in E \vert \LowConsistency} \leq q \cdot \Pr_{R,k}\brackets{\pars{R,k} \in E \vert \LowConsistency} \leq q \cdot \alpha
    \end{equation*}
    and the claim follows.
\end{proof}

\begin{proof}[Proof of \cref{eq:dh_soundness_2}]
    To prove this claim, we will show the following lower bounds:
    \begin{align}
        \Pr\brackets{\Consistent \vert \HighConsistency} &\geq 1 - 2 q \alpha \label{eq:soundness_1},\\
        \Pr\brackets{\oldV^\oldpi = 0  \vert \HighConsistency} &\geq 1 -  \frac{s}{1 - \delta}.\label{eq:soundness_2}
    \end{align}
      Once those are in place, the claim follows by noticing that if $\Consistent$ occurs, then either the new verifier $\newV$ rejects for inconsistency or its queries to $\newpi$ give the same answers as the answers by $\oldV$ to $\oldpi$. Thus, conditioned on $\Consistent$, if $\oldV$ rejects $\oldpi$ on randomness $R$, then so does $\newV$ on $\newpi$. This gives%
    \begin{align*}
        \Pr\brackets{\newV^\newpi = 0 \vert \HighConsistency}
        &\geq
        \Pr\brackets{\Consistent \vert \HighConsistency} \cdot \Pr\brackets{\newV^\newpi = 0 \vert \HighConsistency \cap \Consistent}
        \\
        &\geq
        \Pr\brackets{\Consistent \vert \HighConsistency} \cdot \Pr\brackets{\oldV^\oldpi = 0 \vert \HighConsistency \cap \Consistent}\\
        &=
        \Pr\brackets{\Consistent \cap \oldV^\oldpi = 0 \vert \HighConsistency}
        \\&\geq
        1 - \frac{s}{1 - \delta} - 2 q \alpha.
    \end{align*}
   
    To see \cref{eq:soundness_1}, let $H_{k'}$ denote the event that $i\pars{R,k'}$ has consistency at least $1-2\alpha$ (thus, $H = \bigcap_{k'} H_{k'}$). Then,
    \begin{align*}
        \Pr_{R}\brackets{\overbar{\Consistent} \mid H} &= \Pr_{R} \brackets{\exists k \quad \newpi\pars{R,k} \neq \oldpi\pars{i^{(k)}\pars{R}} \;\Big\vert\; \bigcap_{k'} H_{k'}} \\
        &\leq \sum_{k=1}^{q} \Pr_{R}\brackets{\newpi\pars{R,k} \neq \oldpi\pars{i^{(k)}\pars{R}} \;\Big\vert\; \bigcap_{k'} H_{k'}} \\
        &\leq \sum_{k=1}^{q} \Pr_{R}\brackets{\newpi\pars{R,k} \neq \oldpi\pars{i^{(k)}\pars{R}} \;\Big\vert\; H_k} \\
        &\leq q \cdot 2 \alpha.
    \end{align*}

    For \cref{eq:soundness_2}, recall that the soundness of the original verifier to gives $\Pr\brackets{\oldV^\oldpi = 0} \geq 1 - s$. Then, using elementary probability, we have
    \begin{equation*}
        \Pr\brackets{\oldV^\oldpi = 0 \vert H} \geq \frac{ \Pr\brackets{\oldV^\oldpi = 0} - \Pr\brackets{\LowConsistency} }{\Pr\brackets{\HighConsistency}} \geq \frac{1-s - \delta}{1-\delta} = 1- \frac{s}{1-\delta}.\qedhere
    \end{equation*}
\end{proof}

\subsection{Smoothness}\label{sec:smooth_smoothness}
\newcommand{\closedNhd}[1]{\overbar{\Gamma}\pars{#1}}
For each $R \in \bool^r$ and $k\in \brackets{q}$, let $\closedNhd{R,k}$ denote the closed sampler-neighborhood of $\pars{R,k}$, i.e. the union of the sampler-neighborhood of $\pars{R,k}$ and the singleton containing it $\braces{\pars{R,k}}$. Recall that the sampler is a regular graph with constant degree $\agla-1$, and thus $\agla = \abs{\closedNhd{R,k}}$ for each $R,k$.

Now recall how the new verifier determines its queries (see \cref{alg:new_verifier_simplified}): sample $R\in \bool^r$, and for each $k \in \brackets{q}$ query all the $\agla$ locations in  $\closedNhd{R,k}$. We must now show why this procedure is equally likely to query each location in $\newpi$.

Fix a location $\pars{R^\prime, k^\prime}$ in the new proof. Notice that
\begin{equation}
    \Pr_{\substack{R\in \bool^r\\ k \in \brackets{q}}}\brackets{ \pars{R^\prime,k^\prime} \in \closedNhd{R,k}} = \frac{\agla}{2^r\cdot q},
\end{equation}
where both $R$ and $k$ are distributed uniformly and random.
Recall that the new verifier makes $q\cdot \agla$ queries to the proof, where the $(k,j)$-th query, for $k\in [q]$ and $j\in[\agla]$, is $\closedNhd{R,k}[j]$. We thus get
\begin{equation*}
    \Pr_{\substack{R\in \bool^r\\ k \in \brackets{q}, j\in [\agla]}}\brackets{ \pars{R^\prime,k^\prime} = \closedNhd{R,k}[j]} = \frac{1}{2^r\cdot q}
\end{equation*}
showing that each location $\pars{R^\prime, k^\prime}$ is equally likely to be queried by the new verifier.


\section{Representing an affine function as a low-rank matrix}

\begin{claim} \label{producing_randomness_matrices}
    There is a procedure with the following properties:
    \begin{itemize}
        \item Input: An integer $m$, Boolean vectors $u,v \in \F_2^m$ and a bit $b\in \F_2$.
        \item Output: Two matrices, $A \in (\F_2)^{2^m \times 3}$ and $B\in (\F_2) ^{3 \times 2^m}$, such that $(A\cdot B)_{x,y} = \iprod{x,u} + \iprod{y,v} + b$.
        \item Runtime $\widetilde{O}(2^{m})$.
    \end{itemize}
\end{claim}
\begin{proof}
    We compute $A$ column-by-column.
    \begin{itemize}
        \item The first column is an enumeration of $\iprod{x,u}$, for all $x \in \F_2^{m}$.
        \item The second column is the all-ones vector, denoted by $\Vec{1}$.
        \item The third column is $b\Vec{1}$.
    \end{itemize}
    Next, we compute $B$ row-by-row.
    \begin{itemize}
        \item The first row is $\Vec{1}$.
        \item The second row is an enumeration of $\iprod{y,v}$, for all $y \in \F_2^{m}$.
        \item The third row is $\Vec{1}$
    \end{itemize}
    It is easy to verify that $(A\cdot B)_{x,y} \equiv \iprod{x,u} + \iprod{y,v} + b$, and that the runtime is $O(2^{m} \cdot m) = \widetilde{O}({2^{m}})$.
\end{proof}

\section{The tests and queries of \mytqcite{Section~8.2.1}{BenSassonGHSV2006}} 
\label{appendix:bghsv_query_details}
We list the tests performed by the verifier of \cite[Section 8.2.1]{BenSassonGHSV2006}. Since we only reason about RNL of this verifier, describing the \emph{query pattern} of each test suffices. The reader is referred to \cite[Section 8]{BenSassonGHSV2006} for a complete description (i.e., query pattern and the decision predicate).

	\begin{enumerate}
	\item Let $\lambda \leq \min\{\log n / c, s^3/m^{cm}\}$.

	Construct set $\bset_\lambda\subseteq \braces{1} \times \F^{m-1}$, a $\lambda$-biased set of size at most $O\left(\frac{\log (|\F|^m)}{\lambda}\right)^2$ (see~\cite{AlonGHP1992}).
	\item  Sample a random string $R$ of length $\log(|\bset_\lambda| \cdot |\F|^{m-1}) = (m-1)\log(|\F|) +  O(\log(1/\lambda)) + O(\log\log(|\F|^m))$. Based on this random string, issue queries for the following tests:
	 \begin{enumerate}
	 	\item Robust Edge-Consistency Test$^\Pi(R)$: Use random string $R$ to determine a random $1^{st}$ axis-parallel line in $\F^m$ of the form $\calL = \{(t, a_2, \cdots ,a_{m})\}_{t\in \F}$ for some fixed $a_2, \ldots, a_m \in \F$. Query the oracle $\Pi$ on all points along the line $\calL$ and $S(\calL)$.	
	  \item Robust Zero Propagation Test$^\Pi(R)$: Use random string $R$ to determine a random $1^{st}$ axis-parallel line in $\F^m$ of the form $\calL = \{(t, a_2, \cdots ,a_{m})\}_{t\in \F}$ for some fixed $a_2, \ldots, a_m \in \F$. Query the oracle $\Pi$ on all points along the line $\calL$ and $S(\calL)$.	
	  \item Robust Identity Test$^\Pi(R)$: Use random string $R$ to determine a random $1^{st}$ axis-parallel line in $\F^m$ of the form $\calL = \{(t, a_2, \cdots ,a_{m})\}_{t\in \F}$ for some fixed $a_2, \ldots, a_m \in \F$. Query the oracle $\Pi$ on all points along the line $\calL$.
	  \item Robust Low-Degree Test$^\Pi(R)$: Use random string $R$ to determine a random canonical line $\calL$ in $\F^m$ using the $\lambda$-biased set $\bset_\lambda$. Query the oracle $\Pi$ on all points along the line $\calL$.
	 \end{enumerate}
	\end{enumerate}

\end{document}